\documentclass[preprint,3p]{elsarticle}

\usepackage{etoolbox}
\makeatletter
\patchcmd{\ps@pprintTitle}{\footnotesize\itshape
       Preprint submitted to \ifx\@journal\@empty Elsevier
       \else\@journal\fi\hfill\today}{\relax}{}{}
\makeatother
\usepackage{epsfig,amssymb,amsmath,natbib,listings,lineno,xcolor,natbib,bm,algorithm,url,amsthm,mathtools,psfrag}
\usepackage[noend]{algpseudocode}
\usepackage[normalem]{ulem}
\usepackage{subfig}
\usepackage{float}
\newtheorem{theorem}{Theorem}[section]

\newtheorem{remark}{Remark}[section]

\begin{document}
\begin{frontmatter}

\title{Pass-efficient methods for compression of high-dimensional \\ turbulent flow data}

\author[colorado1]{Alec M. Dunton}
\ead{alec.dunton@colorado.edu}
\author[ctr,upc]{Llu\'is Jofre}
\ead{lluis.jofre@upc.edu}
\author[ctr]{Gianluca Iaccarino}
\ead{jops@stanford.edu}
\author[colorado2]{Alireza Doostan\corref{mycorrespondingauthor}}
\cortext[mycorrespondingauthor]{Corresponding author}
\ead{alireza.doostan@colorado.edu}
\address[colorado1]{Applied Mathematics, University of Colorado, Boulder, CO 80309, USA}
\address[ctr]{Center for Turbulence Research, Stanford University, Stanford, CA 94305, USA}
\address[upc]{Department of Fluid Mechanics, Universitat Polit\`ecnica de Catalunya - BarcelonaTech, Barcelona 08019, Spain}
\address[colorado2]{Smead Aerospace Engineering Sciences, University of Colorado, Boulder, CO 80309, USA}


\begin{abstract}
The future of high-performance computing, specifically on future Exascale computers, will presumably see memory capacity and bandwidth fail to keep pace with data generated, for instance, from massively parallel partial differential equation (PDE) systems.
Current strategies proposed to address this bottleneck entail the omission of large fractions of data, as well as the incorporation of $\textit{in situ}$ compression algorithms to avoid overuse of memory. To ensure that post-processing operations are successful, this must be done in a way that a sufficiently accurate representation of the solution is stored.
Moreover, in situations where the input/output system becomes a bottleneck in analysis, visualization, etc., or the execution of the PDE solver is expensive, the the number of passes made over the data must be minimized. In the interest of addressing this problem, this work focuses on the utility of pass-efficient, parallelizable, low-rank, matrix decomposition methods in compressing high-dimensional simulation data from turbulent flows. A particular emphasis is placed on using coarse representation of the data -- compatible with the PDE discretization grid -- to accelerate the construction of the low-rank factorization. This includes the presentation of a novel single-pass matrix decomposition algorithm for computing the so-called interpolative decomposition. The methods are described extensively and numerical experiments on two turbulent channel flow data are performed.
In the first (unladen) channel flow case, compression factors exceeding $400$ are achieved while maintaining accuracy with respect to first- and second-order flow statistics. In the particle-laden case, compression factors of 100 are achieved and the compressed data is used to recover particle velocities. These results show that these compression methods can enable efficient computation of various quantities of interest in both the carrier and disperse phases.

\end{abstract}

\begin{keyword}
Big-data compression; interpolative decomposition; low-rank approximation; particle-laden turbulence; randomized algorithm; single-pass algorithm
\end{keyword}

\end{frontmatter}


\nolinenumbers
\section{Introduction}	
\label{sec:intro}

The design of modern high-performance-computing (HPC) facilities is constrained by the balance required between financial budget, computing power, and energy consumption.
These constraints force system architects to make difficult trade-offs among supercomputer components, e.g., floating-point performance, memory capacity, interconnect speed, input/output (I/O), etc.
As predicted by Moore's~\cite{Moore1965-A} and Kryder's~\cite{walter2005kryder} laws, many features of supercomputers have improved substantially over the past few decades.
However, memory capacity and bandwidth have failed to keep pace with data generation capabilities. This trend is not reverting and will most likely augment in the near future.
For example, it is expected that the Exascale supercomputers to be deployed during the next decade will provide a 1000-10,000$\times$ increase in floating-point performance but only a 100$\times$ increase in memory availability and access speed~\cite{DoE2012-TR}. 

Flow solvers use random access memory (RAM), I/O, and disk space to store solution states at different times for subsequent restart and post-processing.
As the gap between data generation and storage performance has increased, numerical solvers have typically adapted by saving their state less often, viz. temporal or spatial sub-sampling. This can lead to the loss of important data, rendering it less useful in post-processing operations.
This problem is of particular importance in the case of turbulent flows, as the number of spatial and time integration resolutions required to capture all the flow scales in direct numerical simulation (DNS) increases exponentially with the Reynolds number, $Re$.
Extrapolating this trend to future supercomputing settings, storage subsystems may become considerably underpowered with respect to the number-crunching capacity. In this scenario, the affordable resulting data storage frequency will not be sufficient for conducting meaningful analyses.
A similar problem is encountered in outer-loop studies, such as inference, uncertainty quantification (UQ), and optimization, in which large ensembles of model evaluations for different input values are performed, resulting in a rapid growth of data storage requirements~\cite{Jofre2017-A,Fairbanks2020-A,Jofre2020-A}.
The storage capacity and bandwidth limitations also complicate the applicability of time-decoupled strong recycling turbulence inflow methods~\citep{Wu2017-A}, in which flow data for several characteristic integral times, e.g., eddy-turnover time in homogeneous isotropic turbulence (HIT) or flow through time (FTT) in wall-bounded flows, are stored to disk to be reused later as inflow in spatially developing flow problems. If the prediction described above materializes, flow solvers will need to pursue new strategies in which the data size at each time slice is reduced before writing to disk, a process known as data compression. 

Data compression can be divided into five main categories: lossless, near lossless, lossy, mesh reduction, and derived representations~\cite{lidata}. Focusing on the categories of lossless and lossy compression, the primary contrast between the two is that lossless algorithms guarantee reconstruction of compressed data without any loss of accuracy ---\thinspace within machine precision in the case of near lossless\thinspace--- whereas lossy algorithms do not. Due to their accuracy, lossless compression algorithms are more widely accepted in the scientific community for the purposes of data visualization, analysis, and compression. However, the guarantee of accuracy inherent in these methods comes at the cost of  limited compression ratios\cite{Engelson2000-A,Ratanaworabhan2006-A}. On the other hand, higher compression ratios can be obtained by using lossy data compression algorithms. This comes at the expense that the inverse transformation of the compressed data produces at best an approximation of the original data. 

There are numerous existing methods in truly lossless compression, wherein the reconstructed data is bit-for-bit identical to the original~\cite{lidata}. Examples include the well known method \textit{gzip}~\cite{gailly2003gzip}, as well as entropy-based coders~\cite{shannon2001mathematical,huffman1952method}, dictionary-based coders~\cite{ziv1977universal,ziv1978compression}, and predictive coders, e.g., FPC~\cite{burtscher2009fpc}, FCM~\cite{sazeides1997predictability}, and FPZIP~\cite{lindstrom2006fast}. In a related class of methods, near-lossless compressers, reconstructed data is not identical to the original data due to floating-point round-off errors. Examples from this class of approaches include transformation methods such as lossless Fourier and wavelet transform schemes~\cite{strang1996wavelets}. Because of the limited compression ratio attained by these methods coupled with the disk- and RAM-prohibitive magnitude of the data examined in Section~\ref{sec:results}, lossy compression methods for turbulent flow data are presented as a more appealing alternative. A more in-depth review of these strategies, as well as a more extensive list of sources can be found in~\cite{lidata}. 

As in the case of lossless compression techniques, there are several lossy compression approaches, including~\cite{lehmann2014situ,lehmann2016optimizing,tong2012salient,austin2016parallel,zhao2015time}. Notable lossy compression methods include bit truncation~\cite{lidata, gong2012mloc}, in which simulations are run using 64-bit floating point values but only 32-bit values are saved. Another lossy approach, quantization, entails floating-point values being converted into approximations with smaller cardinality, a fixed size following compression can be achieved, but without guarantees on error~\cite{lidata}. Predictive coding techniques rely on approximating data values using extrapolation from neighboring values. Examples from this class of methods include linear predictors such as \textit{Compvox}~\cite{fowler1994lossless}, SZ~\cite{di2016fast,tao2017significantly,liang2018error}, and \textit{Lorenzo}~\cite{ibarria2003out}, as well as spline-fitting predictors like \textit{Isabela}~\cite{lakshminarasimhan2011compressing,lehmann2014situ}. Transform-based compression methods involve computing the transform of the data, e.g., discrete Legendre transform~\cite{otero2018lossy,marin2016large}, discrete cosine transform (DCT) and wavelet transform, then storing the resulting coefficients in a manner that reduces memory footprint. Within this group of methods is ZFP~\cite{lindstrom2014fixed}, the Karhunen-Lo\'eve transform~\cite{loeve1978probability}, the Tucker decomposition for tensor data~\cite{hitchcock1927expression,tucker1966some,austin2016parallel}, and higher order methods based on it, e.g.,~\cite{vannieuwenhoven2012new,de2000best,kroonenberg1980principal}.

\subsection{Contribution of this work}

This work is concerned with temporal compression of large-scale fluid dynamics simulation data via low-rank matrix decomposition algorithms. Standard low-rank factorization methods for temporal compression of large-scale fluid dynamics simulation data, such as standard proper orthogonal decomposition (POD) or principle component analysis (PCA), are not well suited to high-dimensional data, due to their significant computational cost, high memory requirements, and limited parallel scalability~\cite{HalMarTro2011}. In addition, they are not {\it pass}-efficient in that they need to access/read the data multiple times. Due to the memory bottlenecks inherent in large-scale simulations, methods which minimize the number of passes made over a data matrix are emphasized. 

In this work, we examine the utility of four pass-efficient, low-rank factorization techniques for temporal compression of turbulent flow data. These include a blocked single-pass singular-value decomposition (SBR-SVD)~\cite{yu2017single}, a single-pass and two-pass variants of interpolative decomposition (ID)~\cite{cheng2005compression}. The methods enable low-rank approximation of flows without scalability issues nor bottlenecks in extension to higher dimensions. In building two of the ID methods, we propose using coarse grid (or grid sub-sampled) representations, a.k.a {\it sketch}, of the data (as an alternative to random projections), in order to accelerate the construction of the factorization. This particular choice of sketch enables a single-pass implementation of ID, which to the best of the authors' knowledge, is the first single-pass ID algorithm. The pass-efficiency of these methods becomes crucial when data sizes exceed memory available in RAM, as well as when the process of loading data into RAM becomes a computational bottleneck. In addition, we provide convergence analysis of the ID schemes relying on coarse grid data.

Similar to this work, the authors of~\cite{azaiez2019low} present matrix decomposition methods as an effective compression technique for large-scale simulation data, though pass-efficiency is not emphasized in their work. In~\cite{brand2006fast} and~\cite{zimmermann2018geometric}, the authors present online methods for maintaining a low-rank SVD approximation of simulation data via rank-one updates; this procedure requires significant computation at each step, however. Recent work by Tropp et al.~\cite{tropp2019streaming} addresses this concern, though in their framework matrix updates arrive as sparse or rank-one linear updates; in the single-pass framework presented in this work, updates are assumed to arrive as row vectors into RAM. Moreover, in this work, the temporal (row) dimension of the PDE data matrix does not need to be known {\it a priori}.

As the focus of this paper is achieving temporal compression, methods designed for spatial compression in simulations of turbulent flow such as mesh reduction~\cite{weiss2011simplex} and compressed sensing~\cite{Bourguignon2014-A} are left for future investigation. Another interesting future work is a formal comparison of temporal compression techniques based on low-rank factorization with the aforementioned (non-matrix) techniques, such as SZ and FPZIP. 

The rest of this manuscript is organized as follows.
In Section~\ref{sec:decomposition_algorithms}, low-rank factorization strategies for pass-efficient compression of high-dimensional data, including a novel single-pass ID, are described. Next, numerical results of their compression efficiency and reconstruction accuracy for the computation of flow and particle statistics are discussed in Section~\ref{sec:results}.
Finally, conclusions are drawn and future work is outlined in Section~\ref{sec:conclusions}.

\section{Low-rank decomposition methods for data compression}	\label{sec:decomposition_algorithms}
\subsection{Review of QR and SVD}	
\label{sec:qrsvd}

Let $\bm{A} \in \mathbb{R}^{m \times n}$ denote the data matrix of interest whose rows are time snapshots of flow data, e.g., pressure or velocity, as depicted in Figure~\ref{fig:A_schematic}. While in the below discussions we refer to this matrix, we note that in practice $\bm A$ is not formed explicitly due to the high-dimensionality of data, i.e., large $m$ and/or $n$. Instead rows of $\bm{A}$ are {\it processed} one at a time. The matrix decomposition methods presented in this work involve at their core two canonical decompositions: the QR decomposition and the SVD. QR factorization yields a decomposition of $\bm{A}$ (more precisely $\bm A^T$) of the form $\bm{A}$ = $\bm{QR}$, where $\bm{Q} \in \mathbb{R}^{m \times m}$ is a unitary matrix whose columns form an orthonormal basis for the column space of $\bm{A}$~\citep{Golub}.  This procedure is referred to as the range-finding step~\cite{HalMarTro2011} in the algorithms described later in the paper. The three main approaches for computing this decomposition include pivoted Gram-Schmidt orthonormalization of the columns or rows of $\bm A$, Householder reflections, and Givens rotations~\citep{Golub}. Of particular interest in this work is the rank-revealing QR algorithm, which relies on the full-pivoted Gram-Schmidt procedure~\citep{GuEis1996}. In the execution of this procedure, $k$ pivot columns are selected to form an approximate basis for the range of $\bm{A}$. The selection of these columns induces the concept of a numerical rank~\citep{HonPan1992}. A matrix $\bm{A}$ is said to be of numerical rank $k$ for some $\epsilon > 0$ if there exists a matrix $\bm{A}_k$ of rank $k$ such that $\Vert\bm{A} - \bm{A}_k\Vert \leq \epsilon$. This concept lies at the center of low-rank matrix decomposition methods. 

Also crucial to the compression methods explored is the singular value decomposition (SVD), defined as $\bm{A} = \bm{USV}^{T}$ (the transpose is left unconjugated because the data in all applications is real in this work). The matrices $\bm{U} \in \mathbb{R}^{m \times n}, \bm{V} \in \mathbb{R}^{n \times n}$ are unitary and their columns form orthonormal bases for the column and row spaces of $\bm{A}$, respectively. The matrix $\bm{S} \in \mathbb{R}^{n \times n}$ is a diagonal matrix whose entries are the singular values of $\bm{A}$. To generate low-rank approximations of a matrix $\bm{A}$, one may employ a truncated SVD, which yields a decomposition of the form $\bm{U}_k\bm{S}_k\bm{V}_k^{T}$, with $\bm{U}_k \in \mathbb{R}^{m \times k}, \bm{V}_k \in \mathbb{R}^{n \times k}$ and $\bm{S}_k \in \mathbb{R}^{k \times k}$. In this decomposition, the column spaces of $\bm{U}_k$ and $\bm{V}_k$ are approximations of the $k$-dimensional row and column subspaces of the matrix $\bm{A}$, taken in correspondence to its $k$ largest singular values, which are the entries $\sigma_1,...,\sigma_k$ of the $k \times k$ diagonal matrix $\bm{S}_k$. The product $\bm{D} = \bm{U}_k\bm{S}_k\bm{V}_k^{T}$ forms a rank-$k$ approximation of the matrix $\bm{A}$. By the Eckart-Young theorem~\cite{eckart1936approximation}, a truncated SVD is the theoretically best rank-$k$ approximation of a matrix $\bm{A}$ in Frobenius norm, i.e.,
\begin{equation}
\inf_{\text{rank}(\bm{D}) = k} \Vert\bm{A} - \bm{D}\Vert_F = \Vert\bm{A} - \bm{U}_k\bm{S}_k\bm{V}_k^{T}\Vert_F = \left(\sum_{j=k+1}^{\min(m,n)} \sigma_{j}^2 \right)^{1/2},
\end{equation}
where $\sigma_{j}$ is the $j^{th}$ largest singular value of $\bm{A}$. An analogous statement may be made in terms of the spectral norm (induced 2-norm), in which case the lower bound is
\begin{equation}
\inf_{\text{rank}(\bm{D}) = k} \Vert\bm{A} - \bm{D}\Vert_2 = \Vert\bm{A} - \bm{U}_k\bm{S}_k\bm{V}_k^{T}\Vert_2 = \sigma_{k+1}.
\end{equation}

\begin{figure}[t]
\label{fig:A_schematic}
  \centering
  \includegraphics[trim=80mm 80mm 80mm 70mm, clip,width=0.49\textwidth]{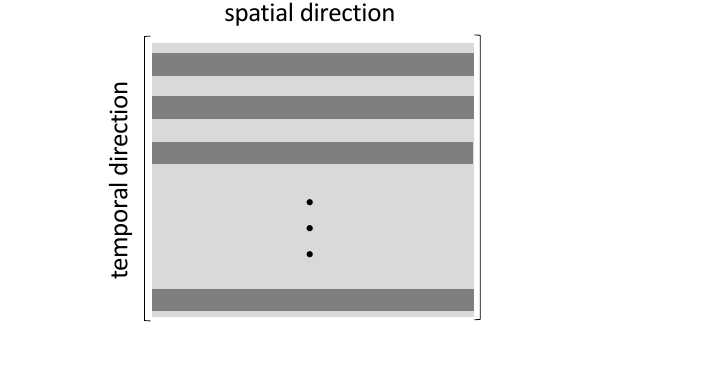}
  \put(-142,10){$m\times n$}
  \put(-68,65){$\longrightarrow\bm{A}$}
  \caption{Schematic of a PDE data matrix $\bm A$ with $m$ time solutions (rows) and $n$ spatial degrees of freedom. }	\label{fig:A_schematic}
\end{figure}

\subsection{Randomized algorithms}
\label{sec:random}

The first methods developed for computing a low-rank approximation of a matrix, e.g., via the SVD, are often computationally expensive, can require $\mathcal{O}(k)$ passes over the input data matrix, lend themselves to limited parallel scalablity, and are not designed to minimize memory movement~\cite{HalMarTro2011}. Moreover, when such schemes are built in a purely deterministic framework, adversarial cases can be introduced, such as those presented by Kahan~\cite{kahan1966numerical}. Developed in order to address these issues, randomized schemes have gained popularity in recent years in low-rank matrix factorizations. These methods rely on embedding the input matrix in a lower dimensional space via a random matrix $\bm\Omega\in\mathbb{R}^{n\times l}$, with $l\ll n$, 
\begin{equation}
\label{eq:sampmatrix}
\bm{A}\bm\Omega,
\end{equation}
referred to as randomized projection~\cite{MarRocTyg2006,liberty2007randomized,HalMarTro2011,Mar2016,yu2017single}. The resulting matrix $\bm{A}\bm\Omega$ is called a sketch matrix~\cite{HalMarTro2011}. The effectiveness of randomized matrix algorithms relies on the utility of the sketch matrix $\bm{A}\bm{\Omega}$. That is, they require that the column space of $\bm{A}\bm{\Omega}$ approximately spans the column space of $\bm{A}$. The theoretical underpinning of random projections in numerical linear algebra is the Johnson-Lindenstrauss Lemma~\cite{johnson1984extensions}, which roughly speaking states that $\mathcal{O}(n)$ points in a Euclidean space  may be randomly embedded in a $\mathcal{O}(\log(n))$-dimensional space such that pairwise distances between points are preserved. This result precipitated the introduction of random matrices as dimension reduction tools in numerical linear algebra~\cite{HalMarTro2011}. Matrices with i.i.d. Gaussian entries -- used exclusively in this work -- are a preeminent examples of such random projections,~\cite{MarRocTyg2006}, though the literature is rife with other techniques for matrix sketching.

Benefits of using randomized methods for generating low-rank approximation, more of which are enumerated in~\citep{Mar2016}, are:
\begin{itemize}
\item The cost of computing a $k$-rank approximation of $\bm{A}$ using deterministic methods, including some of those implemented in this work, requires $\mathcal{O}(mnk)$ operations. By using randomized methods this can be reduced to $\mathcal{O}(mn\log(k) + k^2(m + n))$ or better~\citep{HalMarTro2011}. 

\item Randomized methods require less communication than standard methods, which enables efficient implementation in low-communication environments such as graphics processing units (GPUs)~\citep{MarVor2016}.

\item Of particular interest to this study is that randomized methods allow (in some implementations) for single-pass compression of matrices, which means that the matrix can be compressed as it is streamed and never has to be stored in RAM in its entirety~\cite{HalMarTro2011,yu2017single}.
\end{itemize}


%
%
\subsection{Randomized SVD and single-pass algorithms}
\label{sec:r-svd}

A natural application of random projection is in the construction of SVD. Randomized SVD (R-SVD) methods approximate $\bm{A}$ in the form $\bm{A} \approx \bm{U}\bm{S}\bm{V}^{T}$ for a given target rank $k$ via two main stages. Notice that, for the interest of a simpler notation, we hereafter drop the subscript $k$ from $\bm U$, $\bm V$, and $\bm S$, but keep those in the Algorithms to facilitate their implementation. In the first stage of R-SVD, a basis $\bm{Q}\in\mathbb{R}^{m\times l}$ of the approximate column space of $\bm A$ is identified from the QR factorization of the sketch matrix $\bm{A\Omega}$, where $\bm\Omega\in\mathbb{R}^{n\times l}$ is, e.g., a Gaussian random matrix. Here $l=k+p$, where the so-called over-sampling parameter $p$ is a small number, e.g., $10$ or $20$. In the second stage, the SVD of the smaller matrix $\bm{B}=\bm{Q}^T \bm{A}\approx \tilde{\bm U}\bm{S}\bm{V}^T$ is computed. Recognizing that $\bm{A}\approx\bm{Q}\bm{B}$, the approximate SVD of $\bm{A}$ is given by $\bm{A}\approx(\bm{Q}\tilde{\bm U})\bm{S}\bm{V}^T$, i.e., $\bm{U}=\bm{Q}\tilde{\bm U}$. The details of these steps are described in Algorithm~\ref{alg:randsvd} reported from~\cite{HalMarTro2011}.

\begin{algorithm}[htb]
\caption{Basic Randomized SVD (R-SVD)) $\bm{A} \approx \bm{U}_k \bm{S}_k \bm{V}_k^T$}	\label{alg:randsvd}
\begin{algorithmic}[1]
\Procedure{R-SVD}{$\bm{A}$ $\in \mathbb{R}^{m \times n}$}
\State $k \gets$ target rank
\State $p \gets$ oversampling parameter
\State $l \gets k + p$
\State $\bm{\Omega} \gets randn(n,l)$
\State $\bm{QR} \gets  qr(\bm{A\Omega})$ 
\State $\bm{B} \gets \bm{Q}^T\bm{A}$
\State $\tilde{\bm{U}}$, $\bm{S}$, $\bm{V}
\gets svd(\bm{B})$
\State $\bm{U} \gets \bm{Q}\tilde{\bm{U}}$
\State $\bm{U}_k \gets \bm{U}(:,1:k)$; $\bm{S}_k \gets \bm{S}(1:k,1:k)$; $\bm{V}_k \gets \bm{V}(:,1:k)$
\State $\bm{return}$ $\bm{U}_k$, $\bm{S}_k$, $\bm{V}_k$
\EndProcedure
\end{algorithmic}
\end{algorithm}

More recently, multiple improvements of the R-SVD implementation in Algorithm~\ref{alg:randsvd} have been proposed to address pass and parallel efficiency of R-SVD, which is explained next. Firstly, Algorithm~\ref{alg:randsvd} requires two passes over the data matrix $\bm{A}$ (Steps 6 and 7), which makes it less attractive for the compression of large-scale PDE data that are expensive to store or re-generate. Halko et al.~\cite[Section 5.5]{HalMarTro2011} proposed a single-pass extension of R-SVD as follows: 
\begin{enumerate}[(1)]
\item Generate Gaussian random matrices $\bm{\Omega} \in \mathbb{R}^{n \times l}$ and $\tilde{\bm{\Omega}} \in \mathbb{R}^{m \times l}$, where $k < l \ll \min(n,m)$
\item Compute the products $\bm{A\Omega}$ and $\bm{A}^T\tilde{\bm{\Omega}}$ in a single-pass over $\bm{A}$
\item Using these two products, compute the two QR decompositions  $\bm{A\Omega} = \bm{Q}\bm{R}$ and $\bm{A}^T\tilde{\bm{\Omega}} = \tilde{\bm{Q}}\tilde{\bm{R}}$
\item Solve for the matrix $\bm{B}$, given by the minimum residual solution to the relations $\bm{Q}^{T}(\bm{A} {\bm{\Omega}}) = \bm{B}\tilde{\bm{Q}}^{T}\bm{\Omega}$ and $\tilde{\bm{Q}}^T(\bm{A}^T \tilde{\bm{\Omega}}) = \bm{B}^T \bm{Q}^T\tilde{\bm{\Omega}}$.
\item Compute the SVD of the small matrix $\bm{B}$, yielding $\bm{B} \approx \tilde{\bm{U}}\bm{S}\bm{V}^{T}$
\item Form the matrix $\bm{U}= \bm{Q}\tilde{\bm{U}}$ and set $\bm{U}_k = \bm{U}(:,1:k)$, $\bm{S}_k = \bm{S}(1:k,1:k)$, $\bm{V}_k = \bm{V}(:,1:k)$, to obtain the truncated SVD $\bm{A} \approx \bm{U}_k\bm{S}_k\bm{V}_k^{T}$~\citep{HalMarTro2011}
\end{enumerate}
The main drawback in this method lies in step (3) above, where the typically ill-conditioned matrix $\bm{Q}^{T}\tilde{\bm{\Omega}}$ may lead to considerable accumulation of error compared to the double-pass method presented in Algorithm \ref{alg:randsvd}~\cite{HalMarTro2011}. In order to reduce the communication in the parallel implementation of R-SVD in Algorithm \ref{alg:randsvd} and enable adaptive rank determination, a blocked formulation of the standard algorithm above is proposed in ~\cite{MarVor2016}. In this approach, the orthogonal matrix $\bm{Q} \in \mathbb{R}^{m \times l}$ is separated into $s$ blocks each of size $m \times b$ in the form
\begin{equation}
\bm{Q} = \left[\bm{Q}_1, \bm{Q}_2, \cdots, \bm{Q}_s \right],
\end{equation}
where $s\times b = l$. The computation of the matrix $\bm{Q}$ is then decoupled and carried out on the smaller blocks $\bm{Q}_i$, with all blocks orthogonalized via Gram-Schmidt and concatenated at the end of the process. In doing so, the rank $k$ can be determined so that the factorization admits a prescribed error $\epsilon$. For the interest of clarity and completeness, the main steps of this blocked formulation from \cite{MarVor2016} are now reported, which constitute an alteration to Steps 5-7 of the standard R-SVD in Algorithm \ref{alg:randsvd}. 

\begin{enumerate}[(1)]
\item For each block $i = 1,2,3,...,s$ \textbf{do}
\item Generate Gaussian matrix $\bm{\Omega}_i \in \mathbb{R}^{n \times b}$, where $s\times b = l$ 
\item Compute the QR factorization of $\bm{A\Omega}_i$ to obtain $\bm{Q}_i$
\item Re-orthonormalize $\bm{Q}_i - \sum_{j=1}^{i-1}\bm{Q}_j\bm{Q}_j^{T}\bm{Q}_i$
\item Compute $\bm{B}_i = \bm{Q}^{T}_i\bm{A}$
\item Set $\bm{A} = \bm{A} - \bm{Q}_i\bm{B}_i$
\item If $\Vert \bm{A} \Vert < \epsilon$ $\bm{stop}$
\item Construct 
$\bm{Q} = \left[\bm{Q}_1, \bm{Q}_2, \cdots, \bm{Q}_s \right]$; $\bm{B} = \left[\bm{B}_1^{T}, \bm{B}_2^{T}, \cdots, \bm{B}_s^{T} \right]^{T}$
\end{enumerate}
However, the above implementation of the blocking procedure increases the number of passes through $\bm{A}$ to $\mathcal{O}(s)$. 

In a recent work, Yu et al.~\citep{yu2017single} proposed a single-pass formulation of R-SVD that is shown empirically to result in more accurate low-rank factorizations, as compared to the single-pass R-SVD of \cite{HalMarTro2011}. In more detail, the approach of~\citep{yu2017single} generates an approximate truncated SVD $\bm{A} \approx \bm{U}_k\bm{S}_k\bm{V}_k^{T}$ of rank $k$ following the below steps:
\begin{enumerate}[(1)]
\item Generate Gaussian matrix $\bm{\Omega} \in \mathbb{R}^{n \times l}$, where $k < l \ll n$
\item Obtain the matrices $\bm{Y} = \bm{A\Omega}$ and $\bm{B} = \bm{A}^{T} \bm{Y}$ in a single-pass over $\bm{A}$ (Steps 11-15 in Algorithm~\ref{alg:spsvd})
\item Compute a QR decomposition of $\bm{Y} = \bm{QR}$, set $\bm{B} = \bm{B}\bm{R}^{-1}$ so $\bm{B} \approx \bm{A}^{T} \bm{Q}$ (Steps 16-24 in Algorithm~\ref{alg:spsvd}).
\item Compute the SVD of the small matrix $\bm{B}^{T}  \approx \tilde{\bm{U}}\bm{S}\bm{V}^{T}$ 
\item Construct $\bm{U} = \bm{Q}\tilde{\bm{U}}$ and set $\bm{U}_k = \bm{U}(:,1:k)$, $\bm{S}_k = \bm{S}(1:k,1:k)$, $\bm{V}_k = \bm{V}(:,1:k)$, to extract the truncated SVD $\bm{A} \approx \bm{U}_k\bm{S}_k\bm{V}_k^{T}$ 
\end{enumerate}
As discussed in~\citep{yu2017single}, Step (3) of this implementation can be performed in a blocked and, more importantly, single-pass mode resulting in Algorithm~\ref{alg:spsvd}, herein referred to as single-pass blocked randomized SVD (SBR-SVD) and employed in the numerical examples. 

\begin{algorithm}[t]
\caption{Single-pass Blocked Randomized SVD (SBR-SVD) $\bm{A} \approx \bm{U}_k\bm{S}_k\bm{V}_k^{T}$~\citep{yu2017single}}	\label{alg:spsvd}
\begin{algorithmic}[1]
\Procedure{SBR-SVD}{$\bm{A} \in \mathbb{R}^{m \times n}$}
\State $k \gets$ target rank
\State $p \gets$ over-sampling parameter
\State $l \gets k + p$
\State $b \gets$ block size
\State $s \gets$ number of blocks such that $s\times b = l$
\State instantiate $\bm{Q},\bm{B}$
\State  $\bm{\Omega} \gets randn(n,l)$
\State instantiate $\bm{G}$
\State $\bm{H} \gets zeros(n,l)$
\State $\bm{while}$ $\bm{A}$ is not entirely read through $\bm{do}$
\State $\qquad$ read the next row $\bm{a}$ of $\bm A$
\State $\qquad$ $\bm{g} \gets \bm{a}\bm{\Omega}$ $\quad$ $\bm{G} \gets [\bm{G}; \bm{g}]$
\State $\qquad$ $\bm{H} \gets \bm{H} + \bm{a}^{T}\bm{g}$
\State $\bm{end}$ $\bm{while}$
\State $\bm{for}$ $i = 1, 2,\dots, s$ $\bm{do}$

\State $\qquad$ $\bm{\Omega}_i$ $\gets$ $\bm{\Omega}(:, (i-1)b + 1 : ib)$
\State $\qquad$ $\bm{Y}_i$ $\gets$ $\bm{G}(:, (i-1)b + 1 : ib) - \bm{Q}(\bm{B}\bm{\Omega}_i)$
\State $\qquad$ $\bm{Q}_i,\bm{R}_i$ $\gets$ $qr$($\bm{Y}_i$) 
\State $\qquad$ $\bm{Q}_i,\tilde{\bm{R}}_i$ $\gets$ $qr$($\bm{Q}_i - \bm{Q}(\bm{Q}^{T}\bm{Q}_i))$ 
\State $\qquad$ $\bm{R}_i \gets \tilde{\bm{R}}_i\bm{R}_i$
\State $\qquad$ $\bm{B}_i \gets \bm{R}_i^{-T}(\bm{H}(:,(i-1)b + 1: ib)^{T} - \bm{Y}_i^{T}\bm{QB} - \bm{\Omega}_i^{T} \bm{B}^{T}\bm{B})$
\State $\qquad$ $\bm{Q} \gets \left[ \bm{Q},\bm{Q}_i \right]$ $\bm{B} \gets \left[ \bm{B}^{T},\bm{B}_i^{T} \right]^{T}$
\State $\bm{end}$ $\bm{for}$
\State  $\tilde{\bm{U}}, \bm{S}, \bm{V} \gets svd(\bm{B})$;
\State $\bm{U} \gets \bm{Q}\tilde{\bm{U}}$;
\State $\bm{U}_k \gets \bm{U}(:,1:k)$, $\bm{S}_k \gets \bm{S}(1:k,1:k)$, $\bm{V}_k \gets \bm{V}(:,1:k)$
\State $\bm{return}$ $\bm{U}_k, \bm{S}_k, \bm{V}_k$
\EndProcedure
\end{algorithmic}
\end{algorithm} 

For sufficiently large $p$ (hence $l$), e.g., $p$ between $5$ and $20$~\cite{HalMarTro2011}, the output of this algorithm is an approximately optimal $k$-rank approximation for a matrix per the Eckart-Young theorem~\cite{eckart1936approximation}. More specifically,
\begin{equation}
\mathbb{E}(\Vert \bm{A} - \bm{U}_k\bm{S}_k \bm{V}_k^{T}\Vert_F) \leq \left(1 + \frac{k}{p-1}\right)^{1/2}\left(\sum_{j=k+1}^{\min(m,n)} \sigma_{j}^2 \right)^{1/2}.
\end{equation}
In words, the error is on average only worse than the optimal solution by a factor of $\left[1 + k/(p-1)\right]^{1/2}$~\cite{yu2017single}. 

The SBR-SVD algorithm has computational complexity $\mathcal{O}(mnk)$, which can be reduced to $\mathcal{O}(mn\log k)$ when implemented with certain optimized matrix sketches including the sub-sampled random Fourier transform~\cite{HalMarTro2011,woolfe2008fast,rokhlin2008fast}. In addition, it has an approximate processing storage requirement of $l(m+2n)$ elements in RAM during execution~\citep{yu2017single}. Other single-pass implementations of randomized SVD are available in the literature, see, e.g., ~\cite{boutsidis2016optimal,clarkson2009numerical,tropp2017practical,upadhyay2016fast,woodruff2014sketching}, which are not considered in this work.

\subsection{Interpolative decomposition (ID) and its randomized variant}
\label{sec:ID}

The row ID, a two-pass algorithm, generates a decomposition of a matrix $\bm{A}$ following the form,~\cite{cheng2005compression},
\begin{equation}
\label{eqn:id}
\bm{A} \approx \bm{P}\bm{A}(\mathcal{I}, :) ,
\end{equation}
where the {\it row skeleton} $\bm{A}(\mathcal{I},:) \in \mathbb{R}^{k \times n}$ consists of a set of rows of $\bm{A}$ indexed by $\mathcal{I}\subseteq\{1,\dots,m\}$ with size $\vert\mathcal{I}\vert=k$. Further, $\bm{P} \in \mathbb{R}^{m \times k}$ is a coefficient matrix such that $\bm{P}(\mathcal{I} ,:) = \bm{I}$, with $\bm{I}$ the identity matrix. Row ID earns its name from the fact that it {\it interpolates} $\bm A$ in a basis consisting of a subset of its rows. The core procedure in the algorithm used to generate this decomposition are the column-pivoted (rank-revealing) QR algorithm~\cite{GuEis1996}, which yields the index vector $\mathcal{I}$ along with a least squares problem to compute the coefficient matrix $\bm{P}$. In more detail, first, the rank $k$ column-pivoted QR decomposition of $\bm{A}^T$ is computed. 
\begin{equation}
\bm{A}^T\bm{Z} \approx \bm{QR},
\end{equation}
where $\bm{Z}\in\mathbb{R}^{m\times m}$ is a permutation matrix encoding the pivoting done in the algorithm, $\bm{Q}\in\mathbb{R}^{n\times k}$ has orthonormal columns, and $\bm R\in\mathbb{R}^{k\times m}$ is upper triangular. Separating the matrix $\bm{R}$ by its columns into two sub-matrices $\bm{R} = \left[\bm{R}_{1} \hspace{2pt}\vert\hspace{2pt} \bm{R}_{2} \right]$, where $\bm{R}_{1}\in\mathbb{R}^{k\times k}$ and $\bm{R}_{2}\in\mathbb{R}^{k\times(m-k)}$, and approximating $\bm{R}_{2}\approx\bm{R}_{1}\bm{C}$ yields 
\begin{equation}
\bm{A}^T \approx \bm{Q}\bm{R}_{1} \left[\bm{I} \hspace{2pt}\vert\hspace{2pt} \bm{C}\right]\bm{Z}^{T} 
= \bm{A}^T(:,\mathcal{I})\left[\bm{I} \hspace{2pt}\vert\hspace{2pt} \bm{C} \right]\bm{Z}^{T}
= \bm{A}^T(:,\mathcal{I})\bm{P}^T,
\end{equation}
and hence (\ref{eqn:id}). As shown in~\cite{cheng2005compression}, the rank $k$ row ID of an $m \times n$ matrix features a spectral error bound of
\begin{equation}
\label{eqn:id_bound_orig}
\Vert \bm{A} - \bm{P}\bm{A}(\mathcal{I},:)\Vert_2 \leq \sqrt{1 + k(m-k)}\sigma_{k+1},
\end{equation} 
and the computational complexity of $\mathcal{O}(mnk)$.  Algorithm \ref{alg:id} summarizes the steps involved in ID, which are used in this study. To improve the stability of the QR factorization step, the modified Gram-Schmidt ($mgsqr$) procedure of~\cite{Golub} is employed in Step 3. 

\begin{remark}
\label{rem:IDs}
The ID algorithms presented in this study do not constitute a complete list of methods for generating decomposition of matrices which interpolate a subset of their rows. Other approaches for generating similar decompositions can be found in, e.g.,~\cite{mahoney2009cur,elhamifar2013sparse,dyer2015self}. It is also important to note that in the numerical results presented in Section \ref{sec:results}, the ID refers to a specific variation of the decomposition, the row ID. Analogous definitions of a column ID or double-sided ID are also available in the literature~\cite{Mar2016}.
\end{remark}

\begin{algorithm}[thb!]
\caption{Row ID $\bm{A} \approx \bm{P}\bm{A}(\mathcal{I},:)$~\citep{cheng2005compression}}	\label{alg:id}
\begin{algorithmic}[1]
\Procedure{ID}{$\bm{A}$ $\in \mathbb{R}^{m \times n}$}
\State $k \gets$ approximation rank 
\State $\bm{Q}$, $\bm{R}$, $\mathcal{I} \gets mgsqr(\bm{A}^T,k)$
\State $\bm{C} \gets (\bm{R}(1:k,1:k))^{+}\bm{R}(1:k,(k+1):m)$ $\qquad$ ($^{+}$ denotes pseudo-inverse)
\State $\bm{Z} \gets \bm{I}_m(:,[\mathcal{I}\ ,\ \mathcal{I}^c])$ \hspace{4.2cm} ($\mathcal{I}^c$ is the complement of $\mathcal{I}$ in $\{1,\dots,m\}$)
\State $\bm{P} \gets \bm{Z}\left[\bm{I}_k \hspace{2pt}\vert\hspace{2pt}\bm{C} \right]^{T}$
\State $\bm{return}$ $\mathcal{I}, \bm{P}$
\EndProcedure
\end{algorithmic}
\end{algorithm}

Notice that the row ID uses the entire columns of $\bm A$ -- corresponding to all $n$ degrees-of-freedom of the PDE data --  to find the row skeleton $\bm{A}(\mathcal{I},:)$ and the coefficient matrix $\bm{P}$. Therefore, the computational cost of the QR factorization step depends explicitly on $n$. To reduce the cost of ID, similar random projections of Section \ref{sec:r-svd} have been proposed in \cite{martinsson2011randomized} to generate a sketch of $\bm{A}$ using which an approximate ID of $\bm{A}$ is computed. Specifically, randomized ID (Algorithm \ref{alg:randid}) performs the ID of the sketch matrix $\bm Y=\bm{A}\bm{\Omega}$, with a random matrix $\bm{\Omega}\in\mathbb{R}^{n\times l}$, as
\begin{equation}
\bm{Y}\approx \tilde{\bm{P}}\bm{Y}(\tilde{\mathcal{I}},:).     
\end{equation}
It then uses the row indices $\tilde{\mathcal{I}}$ to set the row skeleton $\bm A(\tilde{\mathcal{I}},:)$ and applies the same coefficient matrix $\tilde{\bm{P}}$; that is,
\begin{equation}
\bm{A}\approx \tilde{\bm{P}}\bm{A}(\tilde{\mathcal{I}},:).     
\end{equation}
As shown in~\cite{martinsson2011randomized}, with high probability depending on $l$, the randomized ID leads to a low-rank factorization of $\bm A$ admitting bounds comparable to that of the standard ID in (\ref{eqn:id_bound_orig}) but with larger constants. 

\begin{algorithm}[t]
\caption{Randomized (Row) ID (Gaussian)~\cite{liberty2007randomized} $\bm{A} \approx \tilde{\bm{P}}\bm{A}(\tilde{\mathcal{I}},:)$}	\label{alg:randid}
\begin{algorithmic}[1]
\Procedure{RandID}{$\bm{A}$ $\in \mathbb{R}^{m \times n}$}
\State $k \gets$ approximation rank 
\State Generate $\bm{\Omega} \in \mathbb{R}^{n \times l}$ with i.i.d. Gaussian entries and $l=k+p$
\State $\bm{Y} \gets \bm{A}\bm{\Omega}$ \label{step:randprojectID}
\State $\tilde{\bm{Q}}$, $\tilde{\bm{R}}$, $\tilde{\mathcal{I}} \gets mgsqr(\bm{Y}^T,k)$
\State $\tilde{\bm{C}} \gets (\tilde{\bm{R}}(1:k,1:k))^{+}\tilde{\bm{R}}(1:k,(k+1):m)$ $\qquad\;\;$ ($^{+}$ denotes pseudo-inverse)
\State $\tilde{\bm{Z}} \gets \bm{I}_m(:,[\tilde{\mathcal{I}}\ ,\ \tilde{\mathcal{I}}^c])$ \hspace{4.4cm} ($\tilde{\mathcal{I}}^c$ is the complement of $\tilde{\mathcal{I}}$ in $\{1,\dots,m\}$)
\State $\tilde{\bm{P}} \gets \tilde{\bm{Z}}\left[\bm{I}_k \hspace{2pt}\vert\hspace{2pt}\tilde{\bm{C}} \right]^{T}$
\State $\bm{return}$ $\tilde{\mathcal{I}}, \tilde{\bm{P}}$
\EndProcedure
\end{algorithmic}
\end{algorithm}

\subsection{Sub-sampled interpolative decomposition}
\label{sec:subid}

Motivated by the randomized ID approach, a faster variant of ID that relies on grid (or data) sub-sampling to generate a deterministic sketch of the original data matrix $\bm{A}$ is proposed. Such a sketch is of particular interest, as it can be generated faster that the random projection -- at least when $\bm\Omega$ is dense -- by relying on the simple observation that for {\it smooth} solutions, a coarse representation of data is able to capture the low-rank subspace of the full solution. 

Let $\mathcal{J}\subseteq\{1,\dots,n\}$, with $k<\vert\mathcal{J}\vert=n_c \ll n$, denote the index of the degrees-of-freedom associated with a coarse subset of the original degrees-of-freedom, i.e. coarse grid or sub-sampled representation of data. The sub-sampled ID generates the rank $k$ ID of the sub-sampled data $\bm{A}_c$,
\begin{equation}
\label{eq:A_c}
\bm{A}_c = \bm{A}(:,\mathcal{J}),
\end{equation}
as 
\begin{equation}
\label{eq:coarseA_id}
\bm{A}_c \approx \hat{\bm{A}}_c =  \bm{P}_c\bm{A}_c(\mathcal{I}_c,:),
\end{equation}
which then induces an ID for the original matrix $\bm{A}$ as
\begin{equation}
\label{eq:coarseregA}
\bm{A} \approx \hat{\bm{A}} = \bm{P}_c \bm{A}(\mathcal{I}_c,:).
\end{equation}
In words, the indices and coefficient matrix obtained from the ID decomposition of the coarse grid data in (\ref{eq:coarseA_id}) are used to generate an interpolation rule for the full data matrix $\bm{A}$, a procedure called lifting in \cite{narayan2014stochastic}. This form of dimension reduction (i.e., directly sub-sampling the columns of the input matrix) is referred to as direct injection in the geometric multi-grid literature~\cite{hackbusch2013multi}. 

In addition to reducing the number of columns to be stored from $n$ to $n_c$, the complexity of computing the coarse-grid ID (\ref{eq:coarseA_id}) is $\mathcal{O}(mn_ck)$ instead of $\mathcal{O}(mnk)$ needed for $\bm{A}$ (see Table~\ref{tab:methods}). Moreover, the sketch time required to obtain $\bm{A}_c$ is $\mathcal{O}(1)$, a far better complexity than that of, e.g., the sub-sampled random Fourier transform ($\mathcal{O}(mn\log(k))$)~\cite{liberty2007randomized,woolfe2008fast}. The main steps of the sub-sampled ID are presented in Algorithm \ref{alg:subsampid}.

In the context of reduced order modeling and uncertainty quantification, the use of coarse grid data to {\it guide} the low-rank approximation of a fine grid quantity of interest has been successfully considered in several recent work; see, e.g., \cite{doostan2007stochastic, narayan2014stochastic,doostan2016bi,fairbanks2017low,hampton2018practical}. One such result is employed from~\cite[Theorem 1]{hampton2018practical} to bound the error of sub-sampled ID.

\begin{theorem}
Let $\bm{A}$ be the original data matrix, $\bm{A}_c$ the sub-sampled (coarsened) matrix as in (\ref{eq:A_c}), $\hat{\bm{A}}_c$ the rank $k$ ID approximation to $\bm{A}_c$ as in (\ref{eq:coarseA_id}), and $\hat{\bm{A}}$ the sub-sampled ID approximation as in (\ref{eq:coarseregA}). For any $\tau \geq 0$, let 
\begin{equation}
    \epsilon(\tau) := \lambda_{\max}(\bm{A}\bm{A}^T - \tau\bm{A}_c\bm{A}^T_c),
    \label{eq:epstau}
\end{equation}
where $\lambda_{\max}$ denotes the largest eigenvalue. Then,
\begin{align}
\label{eqn:id_bound}
\Vert \bm{A} - \hat{\bm{A}}\Vert_2 &\leq \min_{\tau, k \leq \mathrm{rank}(\bm{A}_c)} \rho_k(\tau), \\
\rho_k(\tau) &:= \left[(1 + \Vert\bm{P}_c \Vert_2)\sqrt{\tau\sigma_{k+1}^2 + \epsilon(\tau)} + \Vert \bm{A}_c - \hat{\bm{A}_c} \Vert_2 \sqrt{\tau +  \epsilon(\tau)\sigma_k^{-2}}\right],
\end{align}
where $\sigma_k$ and $\sigma_{k+1}$ are the $k^{th}$ and $(k+1)^{th}$ largest singular values of $\bm{A}_c$, respectively.
\label{thm:subIDerror}
\end{theorem}

The interested reader is referred to~\cite[Theorem 1]{hampton2018practical} for the details of the proof of this theorem. Instead, some remarks regarding the results are provided next. Firstly, following \cite{martinsson2011randomized}, $\Vert\bm{P}_c \Vert_2\le \sqrt{k(m-k)+1}$. The effectiveness of the sub-sampled ID approximation $\hat{\bm{A}}$ depends on the assumptions that $\bm{A}_c$ is low-rank and the optimal $\epsilon(\tau)$ is small. The former assumption follows from the assumption that $\bm A$ is low-rank. To investigate $\epsilon(\tau)$, consider the case in which the physical domain of the PDE is a subset of $\mathbb{R}^3$ and the original data is generated via a uniform grid of size $h$ in each direction. It is also assumed the sub-sampled data corresponds to a coarse subset of this grid with uniform size $H\gg h$ in each direction. It can be shown that, for $\tau^*=n/n_c$,
\begin{equation}
\label{eqn:eps_sub-sampled}
\epsilon(\tau^*)\lesssim H,
\end{equation}
where $\lesssim$ denotes a smaller inequality with a bounded constant.  To see this, we observe that the entry $(i,j)$ of the Gramians $\bm{A}\bm{A}^T/n$ and $ \bm{A}_c\bm{A}^T_c/n_c$ are the approximations of the  Euclidean inner-product of the solution at times $i$ and $j$ via a rectangular (piece-wise constant) rule of size $h$ and $H$, respectively. Therefore, with $\tau^*=n/n_c$, $\Vert\bm{A}\bm{A}^T - \tau^*\bm{A}_c\bm{A}^T_c\Vert_{\max}\lesssim H$, assuming the data snapshots have bounded first derivatives. Then, (\ref{eqn:eps_sub-sampled}) follows given that $\epsilon(\tau)= \Vert\bm{A}\bm{A}^T - \tau\bm{A}_c\bm{A}^T_c\Vert_2\le m\Vert\bm{A}\bm{A}^T - \tau\bm{A}_c\bm{A}^T_c\Vert_{\max}$. When the singular values of $\bm{A}_c$ decay rapidly, together with (\ref{eqn:id_bound}), this dependence of $\epsilon(\tau^*)$ on $H$ suggests the error estimate $\Vert \bm{A} - \hat{\bm{A}}\Vert_2\lesssim H^{1/2}$ for the sub-sampled ID approximation. 

\begin{algorithm}[t]
\caption{Sub-sampled ID $\bm{A} \approx \bm{P}_c \bm{A}(\mathcal{I}_c,:)$}	\label{alg:subsampid}
\begin{algorithmic}[1]
\Procedure{SubID}{$\bm{A}$ $\in \mathbb{R}^{m \times n}$}
\State $k \gets$ approximation rank 
\State $\bm{A}^T_c \gets subsample(\bm{A}^T)$
\State $\bm{Q}_c$, $\bm{R}_c$, $\mathcal{I}_c \gets mgsqr(\bm{A}^T_c,k)$
\State $\bm{C}_c \gets (\bm{R}_c(1:k,1:k))^{+}\bm{R}_c(1:k,(k+1):m)$ $\qquad$ ($^{+}$ denotes pseudo-inverse)
\State $\bm{Z}_c \gets \bm{I}_m(:,[\mathcal{I}_c\ ,\ \mathcal{I}_c^c])$ \hspace{4.4cm} ($\mathcal{I}_c^c$ is the complement of $\mathcal{I}_c$ in $\{1,\dots,m\}$)
\State $\bm{P}_c \gets \bm{Z}_c\left[\bm{I}_k \hspace{2pt}\vert\hspace{2pt}\bm{C}_c \right]^{T}$
\State $\bm{return}$ $\mathcal{I}_c, \bm{P}_c$
\EndProcedure
\end{algorithmic}
\end{algorithm}

\subsection{Single-pass interpolative decomposition}
\label{sec:SPID}

In this section, a simple, single-pass algorithm for generating ID, dubbed single-pass ID (Algorithm~\ref{alg:spid}), is presented. To the best of our knowledge, this is the first single-pass ID algorithm. The standard, randomized, or sub-sampled ID algorithms described in Sections \ref{sec:ID}-\ref{sec:subid} require a second pass through the data to establish the row skeleton matrix $\bm{A}(\mathcal{I}_c,:)$. Instead, in single-pass ID, the row skeleton of the coarse data matrix, i.e., $\bm{A}_c(\mathcal{I}_c,:)$, are interpolated back to the original grid in order to form an approximation to $\bm{A}(\mathcal{I}_c,:)$. Specifically,
\begin{equation}
\label{eqn:single_ID}
\bm{A} \approx \hat{\bm A} = \bm{P}_c\bm{A}_c(\mathcal{I}_c,:) \bm{M},
\end{equation}
where $\bm{P}_c$ and $\bm{A}_c(\mathcal{I}_c,:)$ are as in (\ref{eq:coarseA_id}) and $\bm{M}\in\mathbb{R}^{n_c\times n}$ is a coarse to fine interpolation operator. Such a construction is similar to the prolongation step of multi-grid solvers, and is naturally motivated by the observation that the rows of $\bm{A}_c(\mathcal{I}_c,:)$ are coarse grid representation of those of $\bm{A}(\mathcal{I}_c,:)$. The trade-off for one fewer pass reduces loading and simulation time, but sacrifices accuracy as, depending on the sub-sampling factor $n/n_c$, this method may incur large error when interpolating back onto the fine grid. Therefore, the use of the single-pass ID is justified when rerunning the PDE solver or a second pass through the data to set $\bm{A}(\mathcal{I}_c,:)$ is not desirable. The interpolation step is more computationally expensive than lifting in terms of FLOPs; however, it can be done when the data matrix $\bm A$ is to be reconstructed as opposed to during the PDE simulation or data compression steps.

\begin{algorithm}[htb]
\caption{Single-pass ID $\bm{A} \approx \bm{P}_c\bm{A}_c(\mathcal{I}_c,:) \bm{M}$}	\label{alg:spid}
\begin{algorithmic}[1]
\Procedure{SPID}{$\bm{A}$ $\in \mathbb{R}^{m \times n}, Mesh$}
\State $k \gets$ approximation rank
\State $\bm{A}^T_c \gets subsample(\bm{A}^T)$
\State Form the interpolation matrix $\bm M$ \label{step:buildMmatrix}
\State $\bm{Q}_c$, $\bm{R}_c$, $\mathcal{I}_c \gets mgsqr(\bm{A}^T_c,k)$ 
\State $\bm{C}_c \gets (\bm{R}_c(1:k,1:k))^{+}\bm{R}_c(1:k,(k+1):m)$ $\qquad$ ($^{+}$ denotes pseudo-inverse)
\State $\bm{Z}_c \gets \bm{I}_m(:,[\mathcal{I}_c\ ,\ \mathcal{I}_c^c])$ \hspace{4.4cm} ($\mathcal{I}_c^c$ is the complement of $\mathcal{I}_c$ in $\{1,\dots,m\}$)
\State $\bm{P}_c \gets \bm{Z}_c\left[\bm{I}_k \hspace{2pt}\vert\hspace{2pt}\bm{C}_c \right]^{T}$
\State $\bm{return}$ $\bm{A}_c(\mathcal{I}_c,:)$, $\bm{P}_c$, $\bm{M}$  
\EndProcedure
\end{algorithmic}
\end{algorithm} 

The computational complexity of single-pass ID is the same as that of the sub-sampled ID, i.e., $\mathcal{O}(mn_ck)$. Let $t$ denote the number of coarse grid data points used to construct each of the $n$ fine grid data points. The cost of building and storing the interpolation operator $\bm{M}$, a sparse matrix, is $\mathcal{O}(tn)$, which is negligible relative to any asymptotic dependence on $mn_ck$. In lieu of any spatial compression, the total disk memory required to store the output of the algorithm is far less than the other ID (or SVD) methods described in the prior sections. This is because the coarse skeleton matrix $\bm{A}_c(\mathcal{I}_c,:)$ is stored instead of its fine counterpart. This variant of ID therefore yields a spatio-temporally compression of a data set, as opposed to only temporal compression achieved by the other ID algorithms. Additionally, the total RAM usage of the algorithm is $k(m+n_c) + mn_c + tn$, where $tn$ corresponds to the interpolation matrix $\bm M$. Notice that each step of the single-pass ID is the same as those of sub-sampled ID excluding the interpolation step. The interpolation operator can be represented as a sparse matrix $\bm{M}$ (constructed in Step~\ref{step:buildMmatrix} of Algorithm~\ref{alg:spid}), which can then be stored to disk and applied to the coarse row skeleton $\bm{A}_c(\mathcal{I}_c,:)$ following execution of the algorithm.

A bound on the error of an approximation generated using single-pass ID is presented in the following theorem.
\begin{theorem}
Let $\bm{A}$ be a data matrix, $\bm{A}_c$ the sub-sampled (coarsened) matrix as in (\ref{eq:A_c}), $\bm{M}$ the interpolation operator as in (\ref{eqn:single_ID}) and with associated interpolation error $\bm{E}_I:=\bm{A} - \bm{A}_c\bm{M}$, and $\sigma_{k+1}$ the $(k+1)^{th}$ largest singular value of $\bm{A}_c$. The error of the single-pass ID approximation $\hat{\bm{A}}$ in (\ref{eqn:single_ID}) is bounded as follows
\begin{align}
\Vert \bm{A} - \hat{\bm{A}}\Vert_2 &\leq \Vert \bm{E}_I \Vert_2 + \Vert \bm{M} \Vert_2 \sqrt{1+ k(m-k)}\sigma_{k+1}.
\end{align}
\label{thm:SPIDerror}
\end{theorem}
\begin{proof} Note that,
\begin{equation}
\begin{split}
\label{eq:sperror}
\Vert \bm{A} - \hat{\bm{A}} \Vert_2 &= \Vert \bm{A} - \bm{P}_c\bm{A}_c(\mathcal{I}_c,:)\bm{M} \Vert_2 \\
 &\leq \Vert \bm{A} - \bm{A}_c\bm{M}\Vert_2 + \Vert \bm{A}_c\bm{M} -\bm{P}_c\bm{A}_c(\mathcal{I}_c,:)\bm{M} \Vert_2\\
 &\leq \Vert \bm{E}_I \Vert_2 + \Vert \bm{A}_c - \bm{P}_c\bm{A}_c(\mathcal{I}_c,:) \Vert_2 \Vert \bm{M}\Vert_2  \\
 &\leq \Vert \bm{E}_I \Vert_2 + \Vert \bm{M} \Vert_2 \sqrt{1+ k(m-k)}\sigma_{k+1},
\end{split}
\end{equation}
where the last inequality follows from (\ref{eqn:id_bound_orig}) applied to the ID of $\bm{A}_c$.
\end{proof}

Stated differently, Theorem \ref{thm:SPIDerror} suggests the error of single-pass ID depends on the low-rank structure of $\bm{A}_c$ as well as the error incurred in the interpolation step. When the data is considerably low-rank, the single-pass ID error is dominated by the interpolation error $\Vert\bm E_I\Vert$. Considering piece-wise linear interpolation, when data has bounded second derivative, $\Vert\bm E_I\Vert\lesssim H^2$, where $H$ is the size of the course grid~\cite{brenner2007mathematical}. It is worthwhile to note that the interpolation error is independent of the target rank, instead being a function of $H$, or the sub-sampling parameter $n/n_c$, and the smoothness of the data. 

The sub-sampled ID, on the other hand, is less sensitive to sub-sampling, with an error growing asymptotically like $H^{1/2}$. As the sub-sampling factor is increased, a significantly smaller increase in error when using sub-sampled ID than when using single-pass ID is to be expected. Moreover, the error of sub-sampled ID can be bounded above as a function of the singular values of the coarse matrix $\bm{A}_c$ (see Theorem~\ref{thm:subIDerror}), while the error of single-pass ID saturates once the interpolation error dominates the low-rank approximation error. A consequence of this is that the accuracy of single-pass ID will cease to improve despite increasing the target rank. Consequently, sub-sampled ID outperforms single-pass ID in many cases, as examined in Section~\ref{sec:results}. Despite the general superiority of sub-sampled ID in constructing row-rank factorizations, single-pass ID may achieve greater accuracy in certain cases. In particular, when an insufficiently small approximation rank is used, e.g., $k\sim\mathcal{O}(1)$, the absolute interpolation error from generating $\bm{A}_c(\mathcal{I},:)\bm{M}$ can be negligible compared to the error incurred by the low-rank approximation in and of itself. Moreover, when the coarse data is generated via minimal sub-sampling, e.g., $n/n_c \sim \mathcal{O}(1)$, the interpolation error is small enough to compete with lifting in the reconstruction of $\bm{A}$ in the sub-sampled ID. This is corroborated by the result reported in the left panel of Figure~\ref{fig:1PID} in Section~\ref{sec:flow_compression_efficiency}. 

\begin{remark}
\label{rem:single_id_vs_sbrsvd}
An important difference between single-pass ID and SBR-SVD is how their respective stopping criteria are defined. SBR-SVD requires knowledge of the target rank, i.e., compression factor (CF), whereas ID adaptively evaluates a stopping criterion based on the decay of the singular values of the input matrix via a rank-revealing QR algorithm. SBR-SVD takes as input a matrix, a target rank corresponding to compression dimension, and a block size (see Algorithm~\ref{alg:id}). Consequently, in SBR-SVD the desired compression dimension of a matrix must be known \textit{a priori}. Without extensive knowledge of the data being compressed, the accuracy of a decomposition generated via SBR-SVD cannot be known without revisiting the entire matrix following compression. If the numerical rank is not known in advance, the algorithm must be augmented to make more passes over the input matrix. 
\end{remark}

\subsection{Computational complexity and storage comparison}
\label{sec:cost_comparison}

Table \ref{tab:methods} provides the computational complexity, disk storage, and RAM usage required by the SBR-SVD, full ID, sub-sampled ID, and the single-pass ID algorithm with input matrix of size $m \times n$ and target rank $k$. Recall that $l=k+p$ is slightly larger than the target rank $k$ by the oversampling parameter $p$, e.g., $p$ between $10$ and $20$. In addition, the size of the sub-sampled matrix is considerably smaller than that of the full data matrix, i.e., $n_c\ll n$. $t$ is defined to be the size of the stencil of the interpolation scheme used in to map the coarse grid data onto the fine grid. RAM usage and disk storage are given as the total number of matrix entries stored during and following the execution of the algorithm, respectively. 

\begin{table}[thb]
\centering
\tabcolsep7pt\begin{tabular}{lccc}
\hline
Method & Computational complexity & Disk storage & RAM usage\\
\hline
SBR-SVD & $\mathcal{O}(mnk)$  & $k(m+n)$  & $l(m + 2n)$\\
Full ID & $\mathcal{O}(mnk)$  & $k(m+n)$ & $k(m+n) + mn$ \\
Sub-sampled ID & $\mathcal{O}(mn_ck)$ & $k(m+n)$ & $k(m + n_c) + mn_c$\\
Single-pass ID & $\mathcal{O}(mn_ck)$ & $k(m + n_c) + tn$ & $k(m+n_c) + mn_c + tn$\\
\hline
\end{tabular}
\caption{Computational complexity of the SBR-SVD, full ID, sub-sampled ID and single-pass ID with input matrix of size $m \times n$ and target rank $k$. Variable $l \approx k$ in column 3 is the sub-sampled dimension of the matrix in SBR-SVD. Variable $n_c \ll n$ represents the column dimension of the input matrix after sub-sampling. Here, $t$ is the size of the stencil for the interpolation scheme used in mapping the coarse grid onto the fine grid, e.g., $t=2$ for piecewise linear interpolation in one dimension. Disk storage is given as the total number of matrix entries stored following execution of the algorithm.}	\label{tab:methods}
\end{table}

\section{Numerical experiments}	
\label{sec:results}

The compression efficiency and reconstruction accuracy of SBR-SVD, ID, sub-sampled ID, and single-pass ID methods are investigated in compressing data extracted from the DNS of wall-bounded turbulent flows using the Soleil-MPI low-Mach-number flow solver~\citep{Esmaily2018-A}.
In particular, the canonical periodic channel flow at friction Reynolds number $Re_{\tau} = 180$ is selected for two test cases: single-phase turbulence (Figure~\ref{fig:data_capture}) and particle-laden flow with Stokes numbers $St^{+} = 0, 1, 10$ (Figure~\ref{fig:particle_laden_turbulent_flow}). 
As is customary, $Re_{\tau} = u_{\tau} \delta / \nu$, where $u_{\tau}$ is the friction velocity, $\delta$ is the channel half-height, and $\nu$ is the kinematic viscosity of the fluid; $\nu = \mu / \rho$ with $\mu$ the dynamic viscosity and $\rho$ the density.
The mass flow rate is determined through a mean stream-wise pressure gradient $\langle d p /d x \rangle = -\tau_{w}/\delta$, where $p$ is the pressure and $\tau_{w} = \rho \nu \left( d{\langle u \rangle}/{d y} \right)_{y=0} = \rho u_{\tau}^{2}$ is the wall shear stress, with $\langle u \rangle$ the mean stream-wise velocity.

\begin{figure}[htb]
  \centering
  \includegraphics[width=0.51\textwidth]{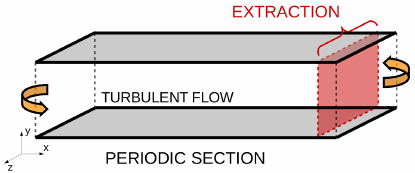}
  \includegraphics[width=0.47\textwidth]{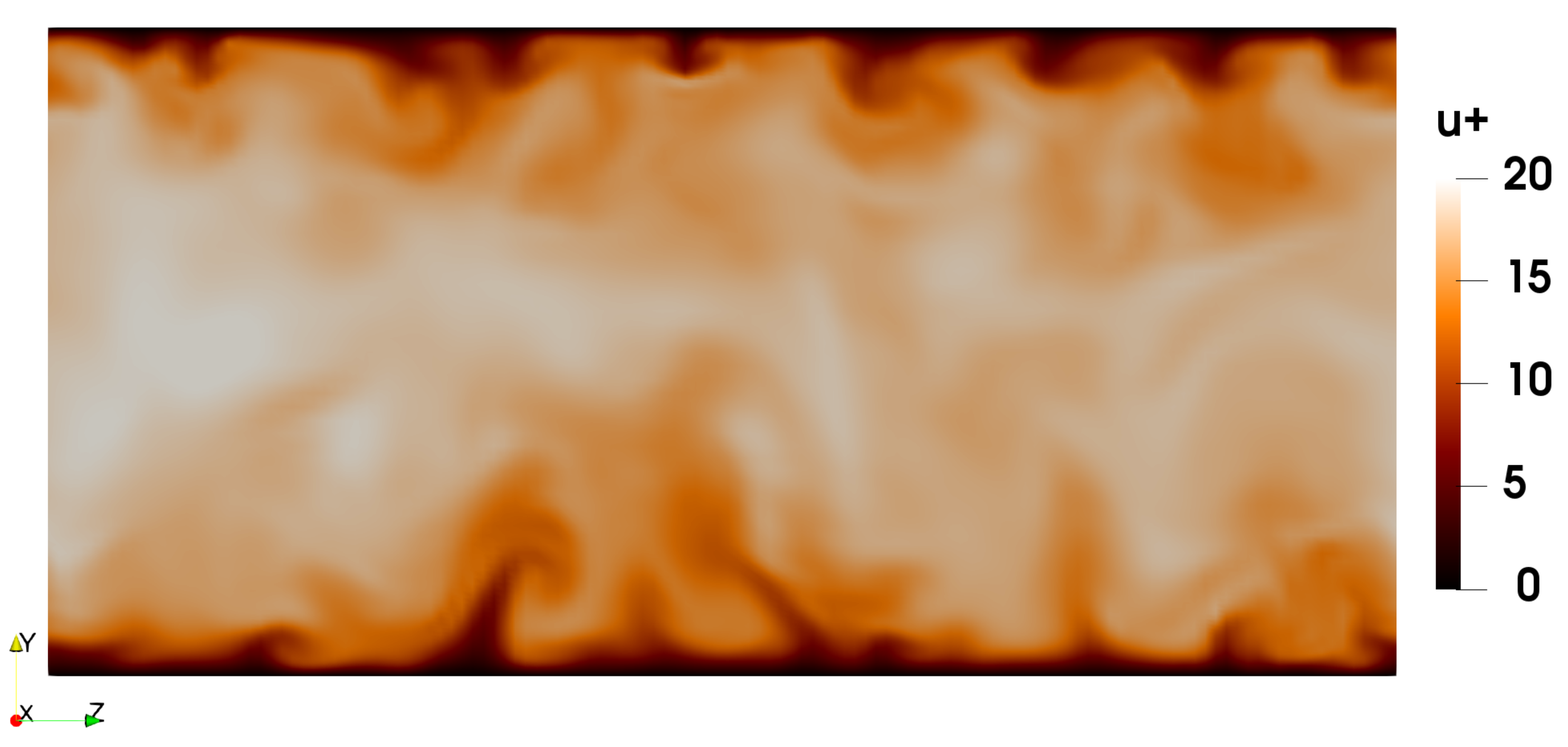}
  \caption{Left: Data extraction computational setup for compression and reconstruction. Right: Stream-wise velocity (wall units) on the extraction plane (instantaneous, span-wise $y$-$z$ plane snapshot).}	\label{fig:data_capture}
\end{figure}

In wall-bounded particle-laden turbulent flows, preferential concentration of the disperse phase ---\thinspace inertial particles are expelled from intense vortical structures and concentrate in regions of the flow dominated by strain\thinspace--- is characterized by the viscous Stokes number $St^{+} = \tau_{p}/\tau_{f}$, defined as the ratio between particle relaxation, $\tau_{p} = \rho_{p}d_{p}^{2}/(18\rho\nu)$ with $\rho_{p}$ the particle density and $d_{p}$ its diameter, and flow, $\tau_{f} = \nu/u_{\tau}^{2}$, time scales.
For the three $St^{+}$ numbers considered, the flow is laden with $200000$ particles resulting in a dilute mixture, i.e., one-way coupling with no particle-particle collisions.
The particle sizes are several orders of magnitude smaller than the smallest (i.e., Kolmogorov) flow scales, and the density ratio between particles and fluid is $\rho_{p}/\rho \gg 1$.
As a result, particles are modeled following a Lagrangian point-particle (PP) approach with Stokes' drag as the most important force~\cite{Maxey1983-A}. A detailed description of the physics modeling and mathematical formulation can be found in~\cite{Jofre2017-A,Fairbanks2020-A,Jofre2020-A}.
The computational domain is $4 \pi \delta \times 2\delta \times 4/3\pi \delta$ in the stream-wise ($x$), vertical ($y$), and span-wise ($z$) directions, respectively.
The stream-wise and span-wise boundaries are set periodic, and no-slip conditions are imposed on the horizontal boundaries ($x$-$z$ planes).
The grid is uniform in the stream-wise and span-wise directions with spacings in wall units equal to $\Delta x^{+} = 9$ and $\Delta z^{+} = 6$, and stretched toward the walls in the vertical direction with the first grid point at $y^{+} = y u_{\tau}/\nu =0.1$ and with resolutions in the range $0.1 < \Delta y^{+} < 8$.
This grid arrangement corresponds to a DNS of size $256 \times 128 \times 128$ grid-points.

The problem under study is illustrated in Figure~\ref{fig:particle_laden_turbulent_flow}.
The simulation strategy starts from a sinusoidal velocity field, seeded with randomly distributed particles, which is advanced in time to reach turbulent steady-state conditions after several FTTs. Based on the bulk velocity, $u_{b} = 1/\delta \int_{0}^{\delta} \langle u \rangle \thinspace d y$, and the length of the channel, $L = 4\pi\delta$, the FTT is defined as $t_{b} = L/u_{b}$.
Once a sufficiently long transient period is surpassed ---\thinspace approximately $10$ eddy-turnover times, $t_{l} \sim \delta / u_{\tau}$\thinspace---the temporal evolution of the velocity field at the outlet plane (130 $\times$ 130 grid: $128$ inner + $2$ boundary points) is extracted for an entire FTT resulting in $25100$ time snapshots.
Similarly, disperse-phase data, viz. time-dependent positions and velocities, for all $200000$ particles are collected over $10000$ time steps.
These two sets of data are utilized to investigate the performance of the compression methods. 
\begin{figure}[t]
  \centering
  \includegraphics[width=0.9\textwidth]{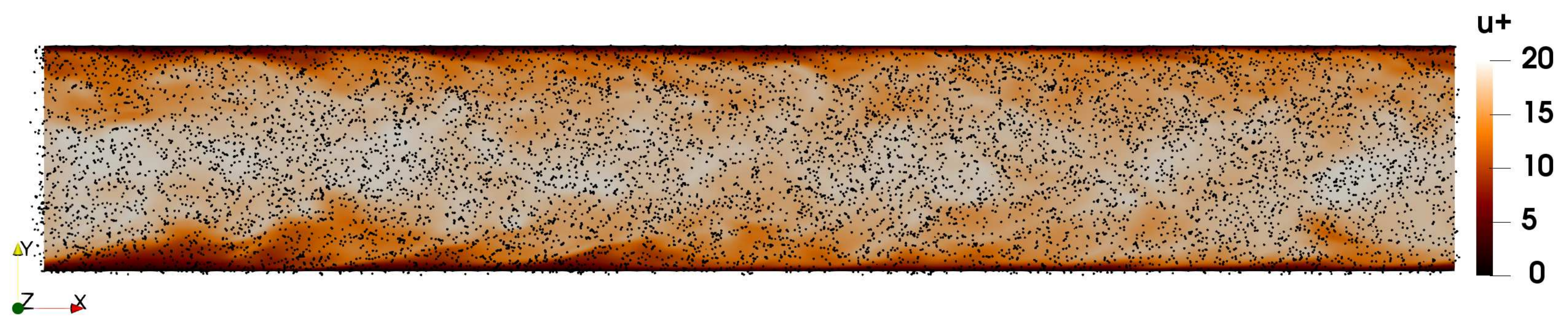}
  \caption{Particle-laden turbulent flow in a periodic channel (instantaneous, stream-wise $x$-$y$ plane snapshot). The quantity represented is stream-wise velocity (wall units) of the fluid phase with particles colored in black. The friction Reynolds number is $Re_{\tau} = 180$ and the viscous Stokes number is $St^{+} = 1$.}	\label{fig:particle_laden_turbulent_flow}
\end{figure}

\begin{figure}[t]
  \centering
  \includegraphics[width=0.49\textwidth]{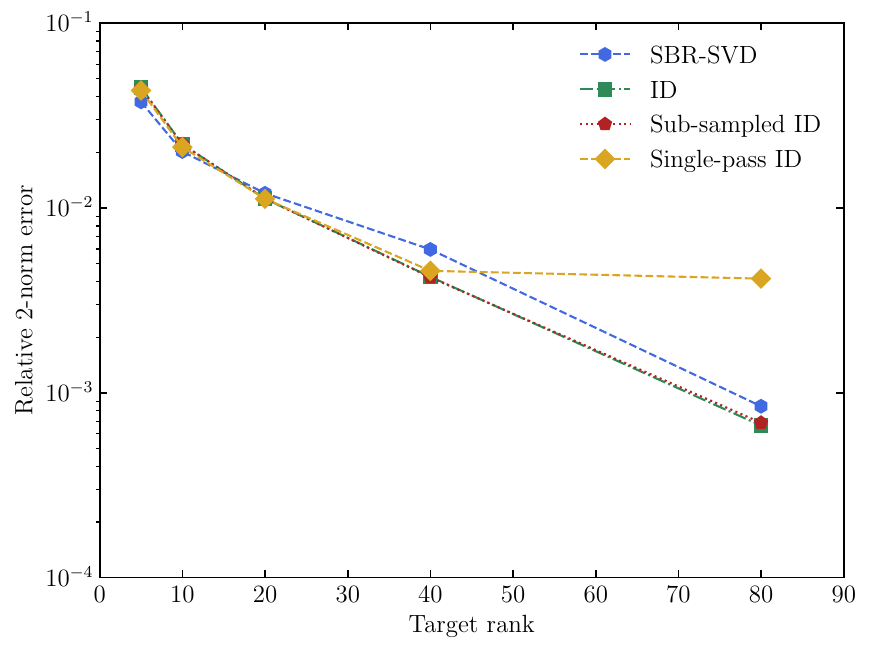}
  \includegraphics[width=0.49\textwidth]{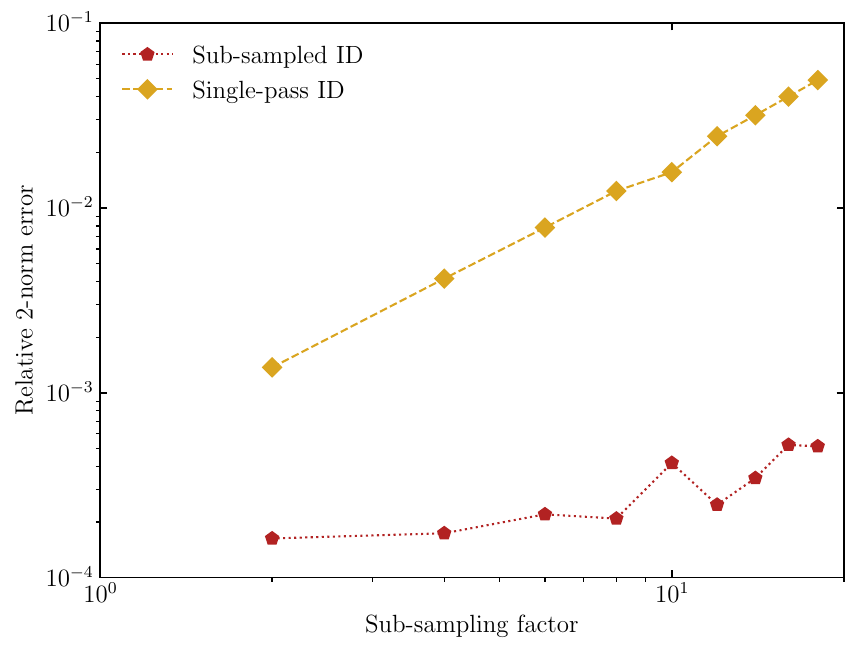}
  \caption{Left: Relative 2-norm error of the four methods for different target ranks for stream-wise velocity data from a turbulent channel flow at $Re_{\tau}=180$. Right: Relative $2$-norm error of single-pass ID and sub-sampled ID with respect to subsampling factor for fixed rank $k=100$.}	\label{fig:1PID}
\end{figure}

\subsection{Test case 1: Outlet flow data compression}	
\label{sec:flow_compression_efficiency}

The $u$ (stream-wise), $v$ (wall-normal), and $w$ (span-wise) velocity fields on the $130 \times 130$ grid slice are captured over $25100$ time steps to form the data to be used in the numerical experiments in this section. 
The four methods utilized for the compression of this flow data are SBR-SVD, ID, sub-sampled ID, and single-pass ID.
In all test cases, the interpolation step in single-pass ID is taken to be a piecewise linear interpolation scheme, though other approaches, e.g., spline interpolation, may be used in its place.

In the left panel of Figure~\ref{fig:1PID}, accuracy results from an experiment performed on  $u$-velocity data collected from the channel flow configuration at 2600 spatial grid points (sampled from the full 16900 in the original data set) over 500 time steps (sampled from the full 25100 in the original data set) are shown. 
The results demonstrate that ID and sub-sampled ID are the most accurate methods as the target rank is increased. SBR-SVD performs almost as well, while the error of single-pass ID saturates after the target rank surpasses 40. This is due to the fact that the interpolation error does not decrease as the target rank is increased. For target rank values 10 and 20, single-pass ID outperforms all other methods, which verifies the claim made in Section~\ref{sec:SPID} that single-pass ID may outperform sub-sampled ID under particular conditions. Th saturation in the single-pass ID results beyond rank $k=40$ is due to the interpolation error $\Vert \bm{E}_I\Vert_2$ in (\ref{thm:SPIDerror}). Due to the nearly indistinguishable performances of the sub-sampled ID and ID algorithms, comparisons against randomized ID are omitted in this test problem. As ID involves no randomization nor sketching of any kind, it is reasonable to expect that any randomized ID algorithm would invariably perform at best the same as ID.

In the right panel of Figure~\ref{fig:1PID}, the errors of single-pass and sub-sampled ID are shown for a range of sub-sampling parameter values. Results clearly indicate that single-pass ID does not match up well to sub-sampled ID in terms of accuracy for large sub-sampling parameter $n/n_c$. The plot is on log-log scale, which also indicates that the rate of convergence of single-pass ID is greater than that of sub-sampled ID with respect to the sub-sampling parameter $n/n_c$. In particular, single-pass ID exhibits a convergence rate of $\mathcal{O}(H^2)$, where $H$ is the grid size used in the piecewise linear interpolation step of single-pass ID.

In Table~\ref{table:runtime}, the speed-up of sub-sampled ID, single-pass ID, and SBR-SVD are reported for different compresion factors, where the speedup of each algorithm is defined as
\begin{equation}
    \text{speedup} := \frac{\text{runtime of slowest algorithm}}{\text{runtime of algorithm}} ,
\end{equation} 
and the compression factor is generally defined as
\begin{equation}
    \text{compression factor} := \frac{\text{number of matrix entries in original data}}{\text{number of matrix entries in compressed data}} ,
\end{equation} 
while the temporal compression factor reported in Table~\ref{table:runtime} is 
\begin{equation}
    \text{temporal compression factor} := \frac{\text{number of rows in original data}}{\text{target rank of approximation}} .
\end{equation} 

Examining the table, sub-sampled ID and single-pass ID (without interpolation) are by far the fastest of the methods considered; they compresses a matrix roughly $60\times$ faster than ID and almost twice as fast as SBR-SVD in the lower compression regimes. Single-pass ID with the interpolation step included, i.e., multiplication by $\bm M$, offers the smallest speedup in all three cases. However, in the context of {\it in situ} compression of simulation data, the interpolation step is an {\it a posteriori} operation in that it is needed when reconstructing the data. If the interpolation step is ignored in this study, the runtimes of the sub-sampled ID and single-pass ID are identical. While the cost of the runtimes of single-pass ID (without interpolation) and sub-sample ID are the same, we remind that sub-sampled ID require an entire simulation to be run again to extract fine-grid row-skeleton snapshots. Our experiment does not take this into account; if a simulation takes days to complete while interpolation takes minutes, the computational savings of single-pass ID over sub-sampled ID will be significant. Note that the difference in performance time between the four methods increases as the compression factor is decreased. 

\begin{figure}[t]
  \centering
  \includegraphics[width=0.49\textwidth]{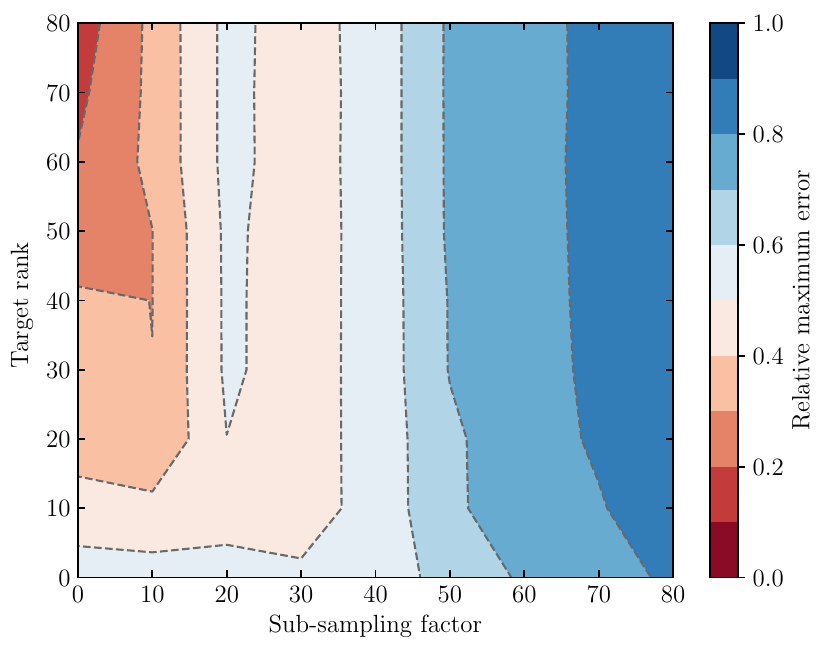}
  \includegraphics[width=0.49\textwidth]{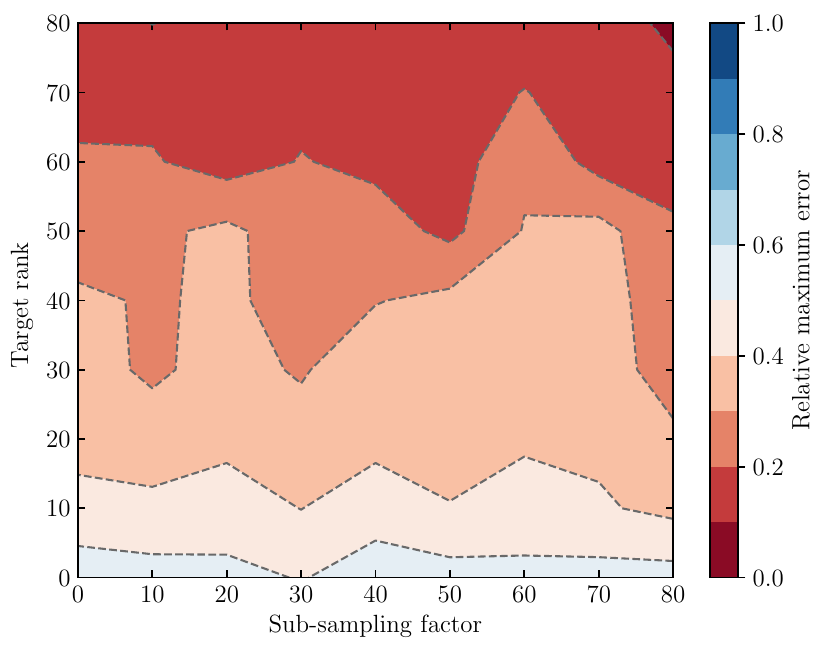}
  \caption{Relative maximum error as function of sub-sampling factor and target rank. Data corresponds to outflow stream-wise velocity, $u$, extracted from channel flow at $Re_{\tau}=180$. Left: Single-pass ID. Right: Sub-sampled ID.}	\label{fig:errors}
\end{figure}

\begin{figure}[t]
  \centering
  \includegraphics[width=0.48\textwidth]{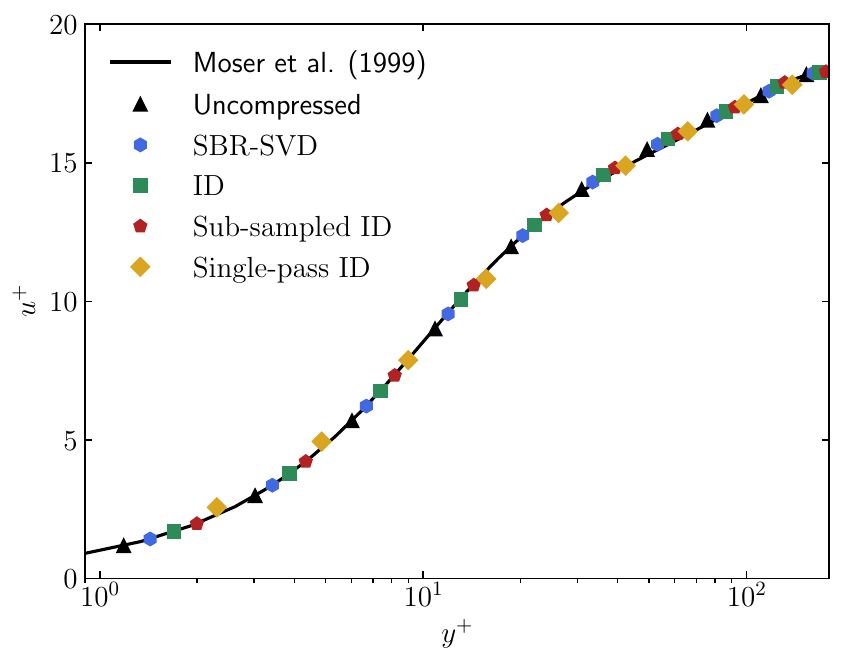}
  \includegraphics[width=0.497\textwidth]{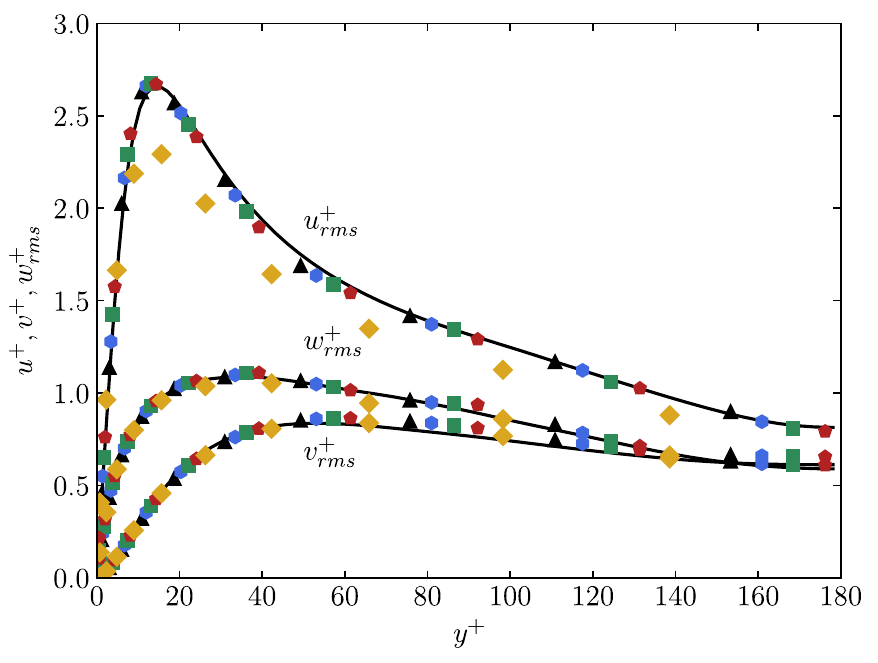}
  \includegraphics[width=0.48\textwidth]{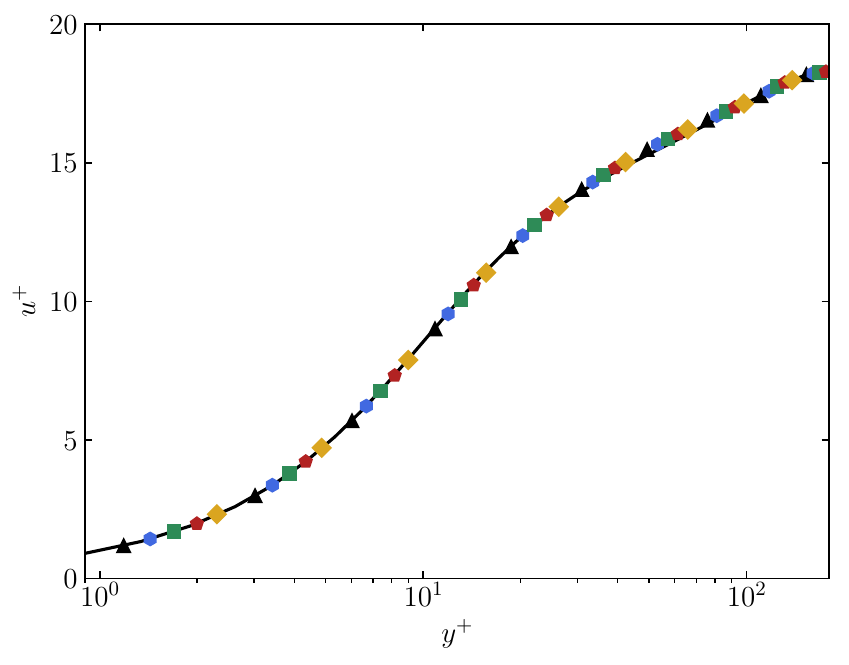}
  \includegraphics[width=0.497\textwidth]{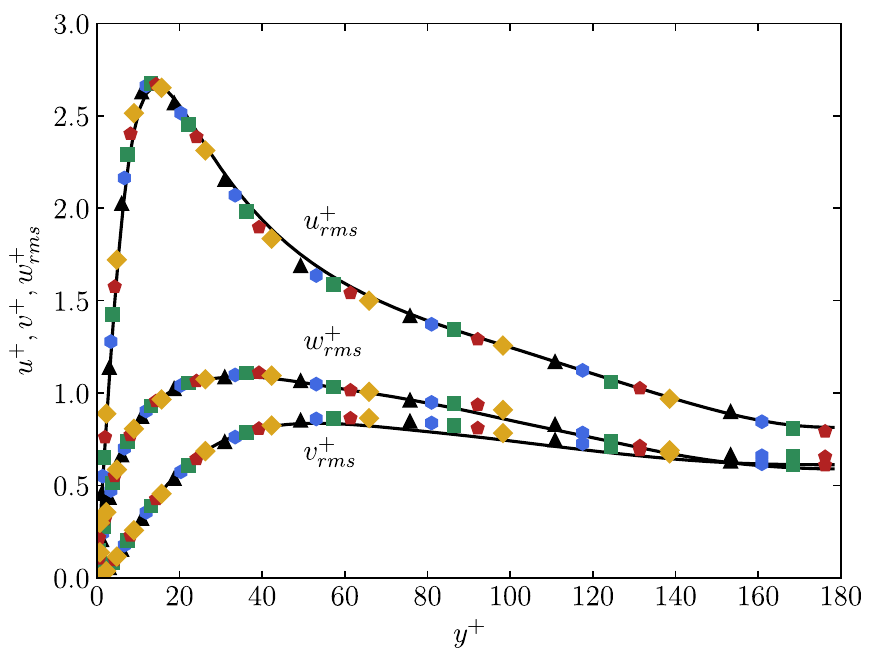}	
  \caption{Left column: Mean stream-wise velocity profile (wall units). Right column: Root mean square velocity fluctuations (wall units). Top row: sub-sampling parameter of 12. Bottom row: sub-sampling parameter of 4. Moser et al. DNS~\cite{Moser1999-A} (black solid lines), uncompressed (black triangles), SBR-SVD (blue hexagons), full ID (green squares), sub-sampled ID (red pentagons) and single-pass ID (yellow diamonds).}	\label{fig:particle_unladen_qoi_1}
\end{figure}

\begin{table}[t]
{\small
\centering
\tabcolsep7pt\begin{tabular}{cccc}
\hline
\small{Temporal compression factor ($m/k$)} & \small{SBR-SVD} & \small{Sub-sampled ID}  & \small{Single-pass ID without (and with) interpolation}\\
\hline
$200$ & $8.3$  & $19.6$ & $19.6$ ($5.4$) \\
$100$ & $16.4$ & $41.7$ & $41.7$ ($8.1$) \\
$20$  & $30.3$ & $55.6$ & $55.6$ ($13.2$) \\
\hline
\end{tabular}
\caption{Speedup in runtime using the SBR-SVD, sub-sampled ID, and single-pass ID relative to that of the full ID algorithm for compressing outlet stream-wise velocity data from channel flow at $Re_{\tau}=180$. The cost of second pass over the data is ignored for Sub-sampled ID. The results reported in the parentheses include the interpolation cost of single-pass ID needed only for data reconstruction.}	\label{table:runtime}
}
\end{table}
\begin{table}[t]
{\small
\centering
\tabcolsep7pt\begin{tabular}{ccccccc}
\hline
 & Analytical & Uncompressed & SBR-SVD & Full ID & Sub-sampled ID & Single-pass ID \\
\hline
$Re_{\tau}$ & $180$            & 180.3  & 180.4  & 180.3  & 180.3 & 180.3\\
$Re_{b}$    & $\approx 5600$   & 5569.4 & 5569.5 & 5569.4 & 5569.4 & 5569.4 \\
$C_{f}$     & $\approx 8.2\cdot10^{-3}$ & $8.3\cdot10^{-3}$ & $8.3\cdot10^{-3}$ & $8.3\cdot10^{-3}$ & $8.3\cdot10^{-3}$ & $8.3\cdot10^{-3}$ \\
\hline
\end{tabular}
\caption{Friction, $Re_{\tau}$, and bulk, $Re_{b}$, Reynolds numbers, and skin-friction coefficient, $C_{f}$, of channel flow at $Re_{\tau}=180$ obtained from uncompressed and compressed data. The compressed data has a compression factor of approximately 150 for all the methods considered.}	\label{tab:Reynolds_numbers_skin_friction}
}
\end{table}

In Figure~\ref{fig:errors}, the relationship between the sub-sampling parameter, target rank, and relative entry-wise error of single-pass ID and sub-sampled ID is shown. For a data matrix $\bm{A}$ and its approximation $\hat{\bm{A}}$, the relative entry-wise error is defined as
\begin{equation}
    \max_{i,j} \frac{\vert \bm{A}_{i,j} - \hat{\bm{A}}_{i,j} \vert}{\vert \bm{A}_{i,j}\vert}.
\end{equation}

The data matrix selected for the test is a $5000\times 16900$ $u$-velocity outflow data obtained via sub-sampling time realization of the full-scale channel flow data. The plots show that the error of the single-pass method is heavily dependent on the sub-sampling parameter, which is reflected in the vertical stripe contours in the plot. The error of the sub-sampled ID, on the other hand, depends primarily on the target rank of the approximation, viz. the horizontal stripe contours in the right panel of Figure~\ref{fig:errors}. 

To investigate the impact of the accuracy loss, we use the compressed data (with 25100 time slices) to generate inflow velocity and calculate time- and space-averaged ($x$ and $z$ directions) first- and second-order flow statistics. The quantities of interest considered are the numerical friction, $Re_{\tau}$, bulk Reynolds number, $Re_{b} = 2 u_{b} \delta/\nu$, skin-friction coefficient, $C_{f} = \tau_{w}/(1/2 \rho u_{b}^{2})$, the mean stream-wise velocity profile, $u^{+} = \langle u \rangle/u_{\tau}$, and the root mean square (rms) velocity fluctuations, $u_{rms}^{+}$, $v_{rms}^{+}$, and $w_{rms}^{+}$.
Reference solutions for all these QoIs are available in the literature. 
For instance, $Re_{b}$ can be analytically approximated from $Re_{\tau} \approx 0.09 Re_{b}^{0.88}$~\citep{Pope2000-B}.
Once $Re_{b}$ is known, $C_{f}$ is directly obtained by calculating $u_{b}$ from the bulk Reynolds number definition and noticing that $\tau_{w}$ and $\rho$ depend on $\langle d p /d x \rangle$ and $Re_{\tau}$.
Regarding mean and fluctuation velocity profiles, reference DNS data are widely available, for example, from Moser et al.~\cite{Moser1999-A}.

\begin{figure}
  \centering
  \includegraphics[width=0.48\textwidth]{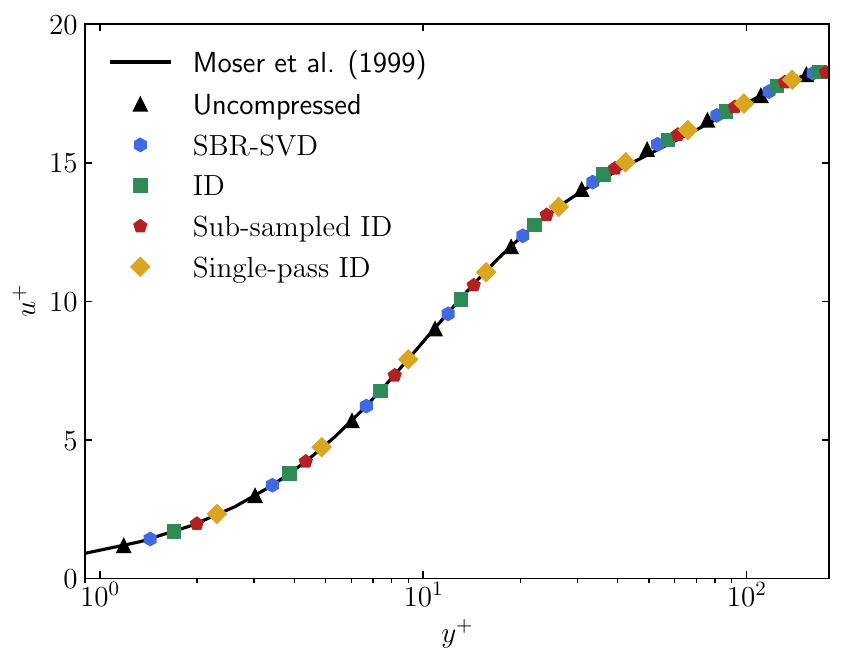}
  \includegraphics[width=0.497\textwidth]{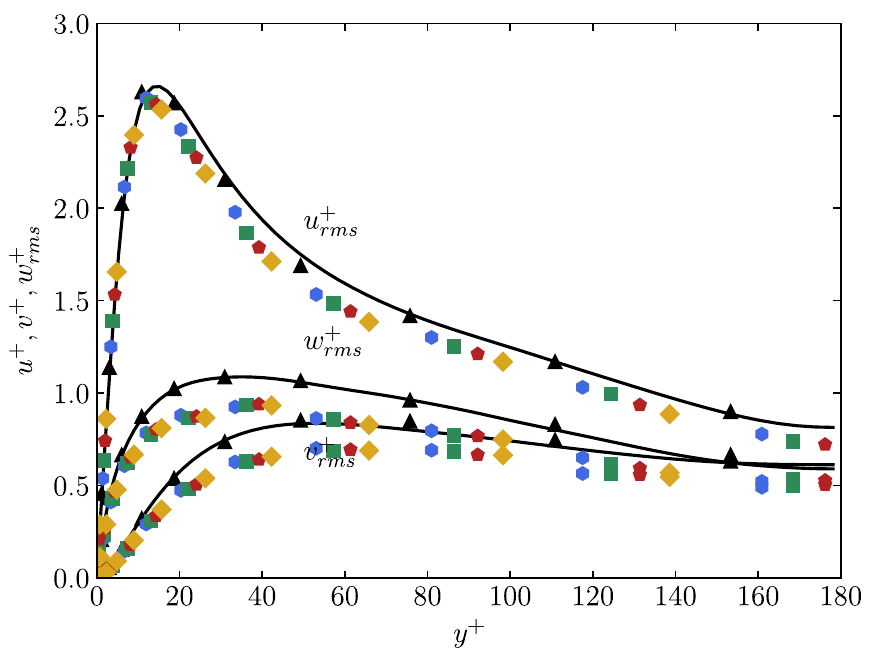}
  \includegraphics[width=0.48\textwidth]{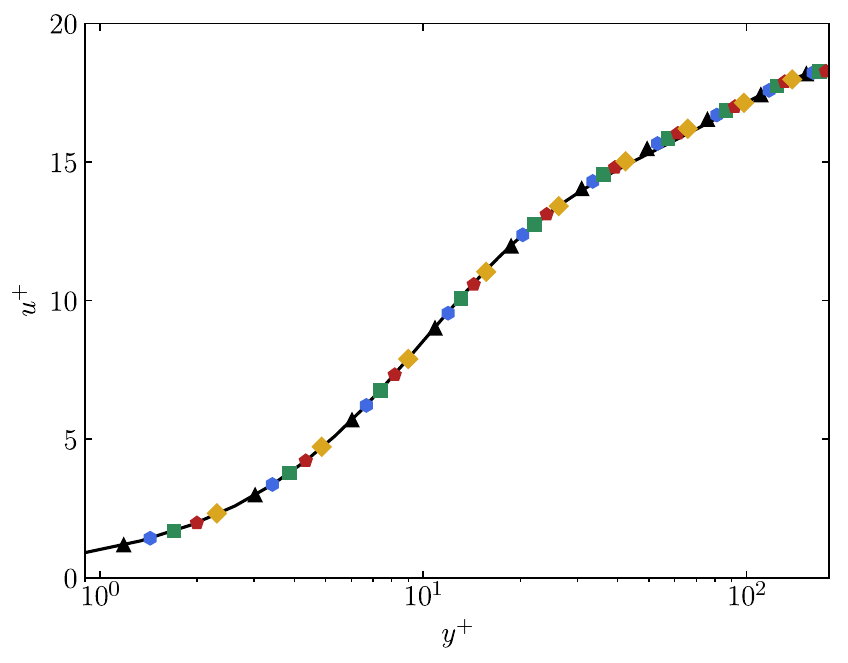}
  \includegraphics[width=0.497\textwidth]{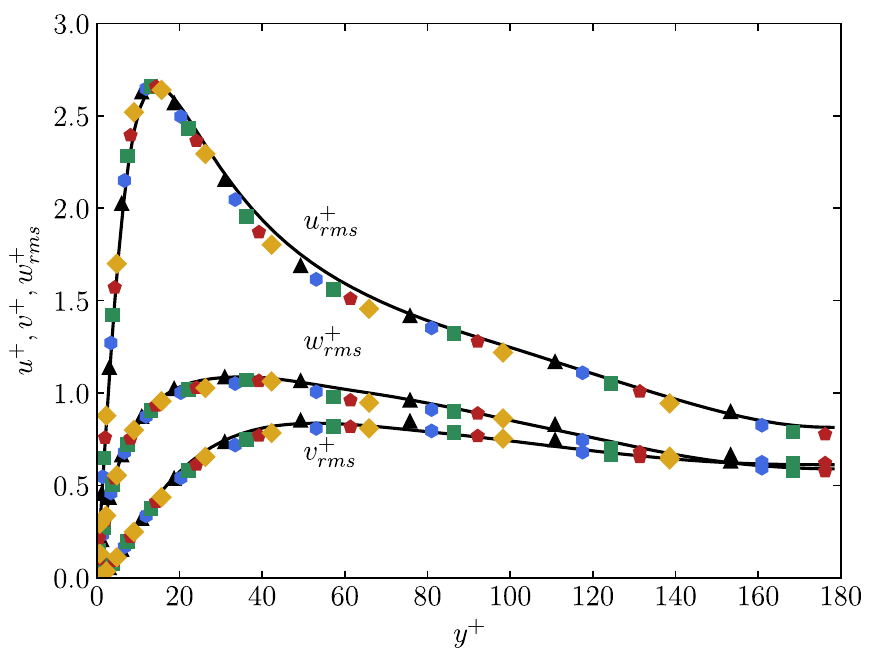}
  \includegraphics[width=0.48\textwidth]{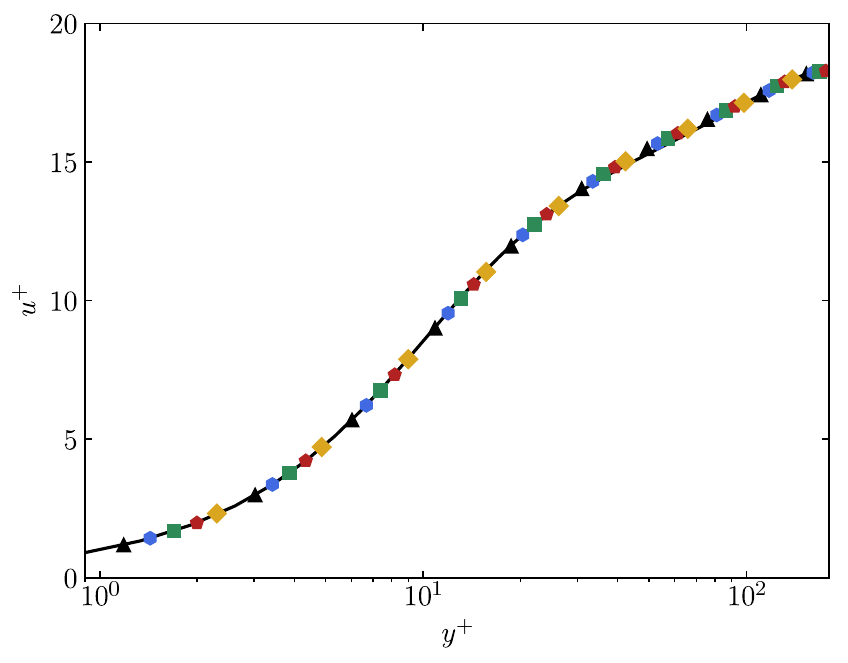}
  \includegraphics[width=0.497\textwidth]{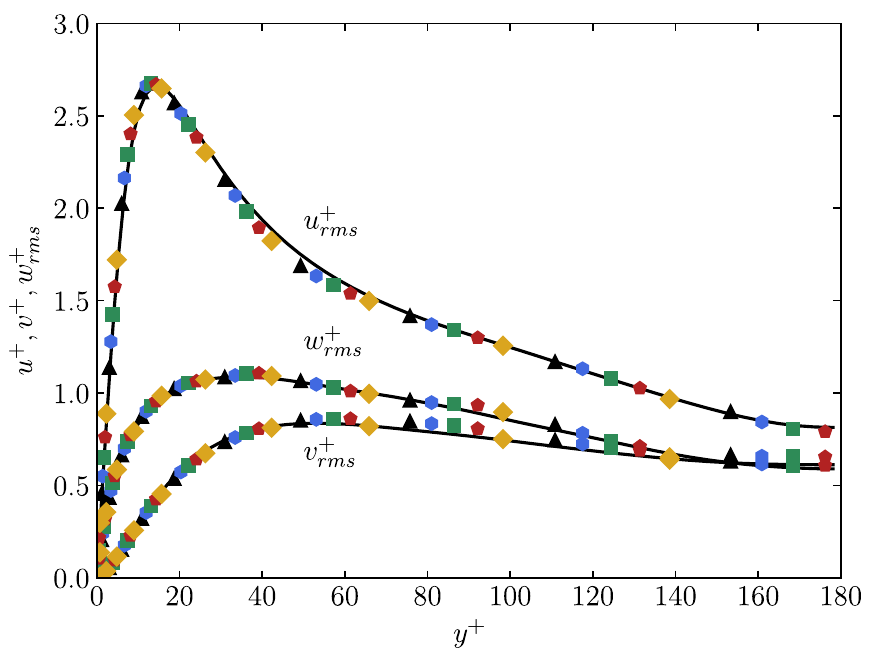}
  \caption{Left column: Mean stream-wise velocity profile (wall units). Right column: Root mean square velocity fluctuations (wall units). Top row: target rank of 20 (CF = 840). Center row: target rank of 40 (CF = 420). Bottom row: target rank of 70 (CF = 240). Moser et al. DNS~\cite{Moser1999-A} (black solid lines), uncompressed (black triangles), SBR-SVD (blue hexagons), full ID (green squares), sub-sampled ID (red pentagons) and single-pass ID (yellow diamonds).}	\label{fig:particle_unladen_qoi_2}
\end{figure} 

Values computed from compressed data, as well as reference and uncompressed solutions, are shown in Table~\ref{tab:Reynolds_numbers_skin_friction} and Figures~\ref{fig:particle_unladen_qoi_1}~and~\ref{fig:particle_unladen_qoi_2}, with the compressed data requiring approximately 150$\times$ less storage than the original. Examining the results shown in Table~\ref{tab:Reynolds_numbers_skin_friction} and Figures~\ref{fig:particle_unladen_qoi_1}~and~\ref{fig:particle_unladen_qoi_2}, the first observation is that the uncompressed and compressed results agree well with the analytical and DNS reference data.
The second observation is that the four strategies considered provide similar compression accuracies ---\thinspace evaluated through the reconstructed velocity field\thinspace--- and are in agreement with the uncompressed data values.
In particular, the relative errors for the quantities in Table~\ref{tab:Reynolds_numbers_skin_friction} are below $6\cdot 10^{-3} \thinspace \%$, and the mean stream-wise velocity and rms fluctuations accurately reconstruct the uncompressed data solution. These results indicate that in applications involving spatially developing turbulent wall-bounded flow, this compressed data can be reused as inflow in subsequent simulations. The disk memory required to store inflows in these simulations can be reduced by factors exceeding 100, thereby placing less strain on the memory and I/O of the system in use. 

In Figure~\ref{fig:particle_unladen_qoi_1}, the reconstruction of the mean stream-wise velocity and root-mean-squared (rms) velocity fluctuations using the four methods described in Section~\ref{sec:decomposition_algorithms} is shown. In the top two panels a sub-sampling factor of 12 in single-pass ID is used, whereas the bottom two panels correspond to a sub-sampling factor of 4. The results clearly demonstrate that the mean stream-wise velocity profile can be recovered in light of substantial interpolation error in single-pass ID. On the other hand, the rms velocity fluctuations, specifically $u_{rms}^+$, are much more sensitive to the reconstruction accuracy achieved by the methods. This highlights the delicate balance between interpolation error and computational performance inherent to the single-pass ID, as well as its application dependence.

In Figure~\ref{fig:particle_unladen_qoi_2} the same statistics as those presented in Figure~\ref{fig:particle_unladen_qoi_1} are shown for three different compression factor values. In this series of tests, the sub-sampling parameter used in the single-pass ID is set to 4 in all cases, as is presented in the bottom panels of Figure~\ref{fig:particle_unladen_qoi_1}. The compression factor values chosen (i.e., target rank) are 240 (bottom panels), 420 (middle panels) and 840 (top panels). The experiment demonstrates that all four methods can recover first- and second-order statistics to within reasonable accuracy while achieving compression factors exceeding 400. Beyond these values, the accuracy of the statistics recovered from the compressed data begins to drop off significantly. Noteworthy is the performance of the single-pass ID, which performs as well as the other three methods presented, indicating its utility as an online data compression algorithm.

\begin{figure}[t]
  \centering
  \includegraphics[width=0.475\textwidth]{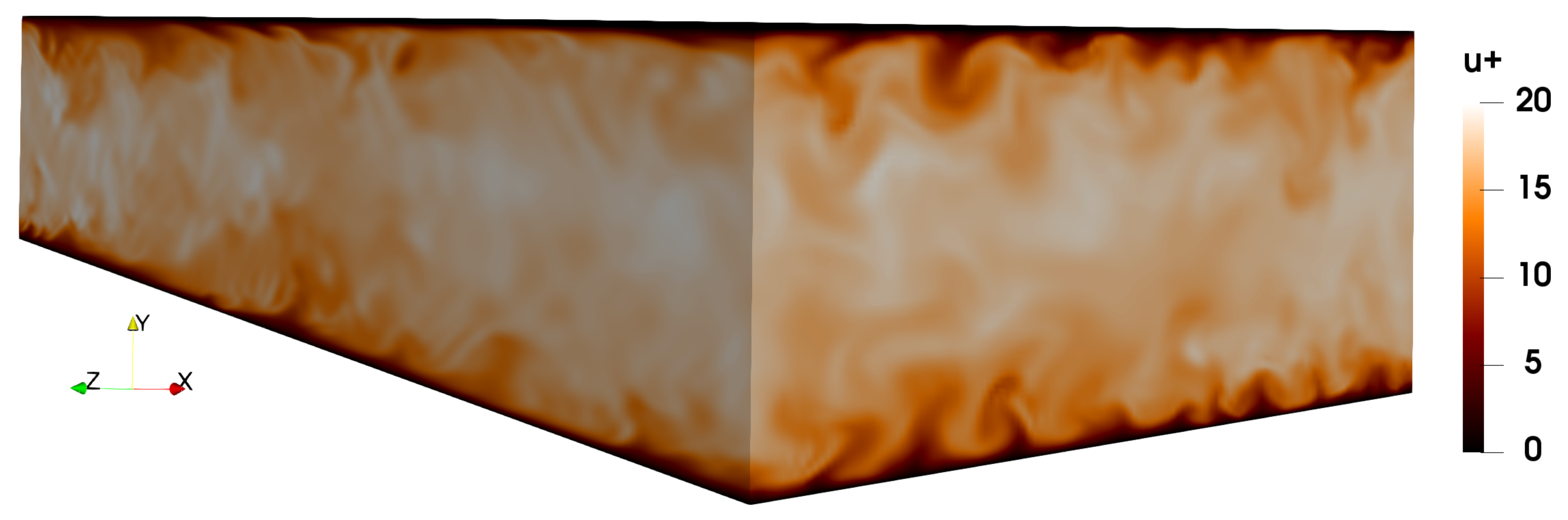}
  \hspace{2.5mm}
  \includegraphics[width=0.475\textwidth]{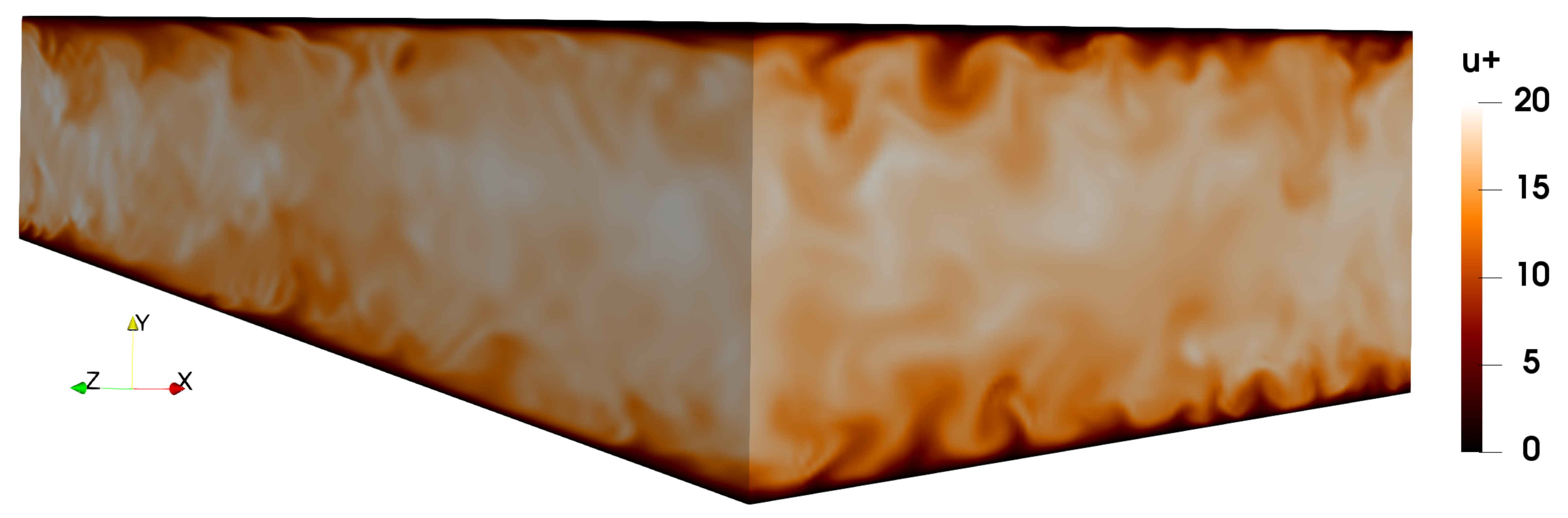}

  \vspace{5.0mm}
  
  \includegraphics[width=0.475\textwidth]{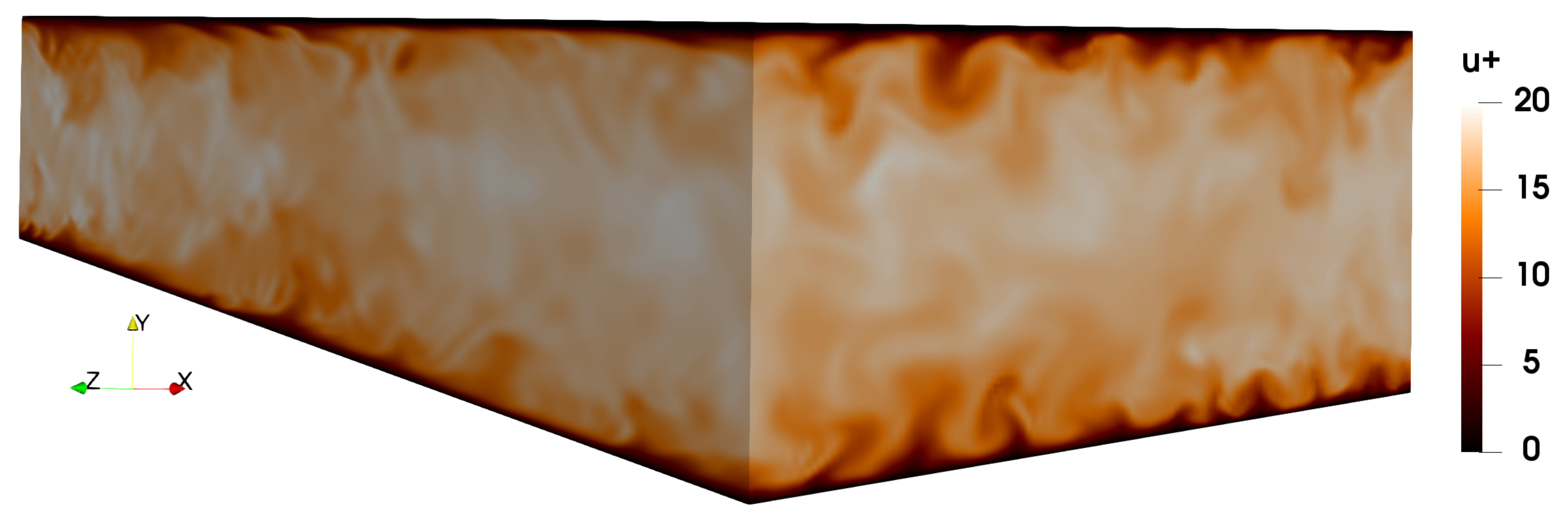}
  \hspace{2.5mm}  
  \includegraphics[width=0.475\textwidth]{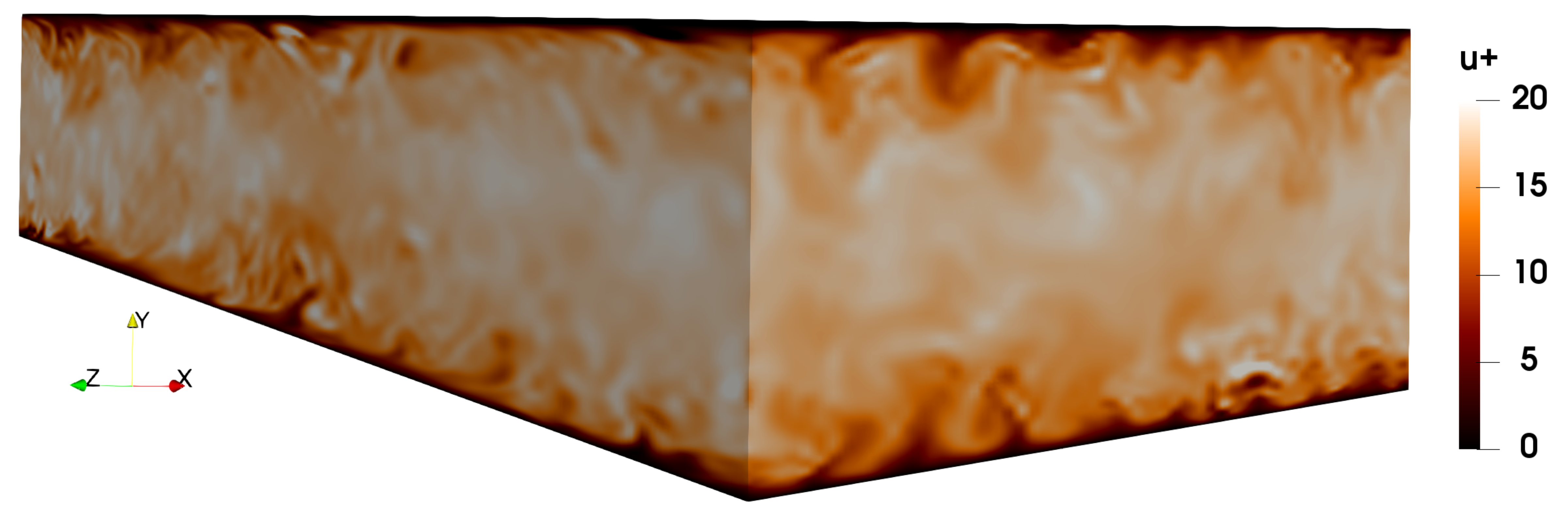}
  \vspace{-2mm}
  \caption{Instantaneous snapshots of channel flow stream-wise velocity $u$ from the uncompressed data (top left), reconstructed using SBR-SVD (top right), reconstructed using sub-sampled ID (bottom left), and reconstructed using single-pass ID (bottom right).}	\label{fig:flowreconstruct}
\end{figure}

\subsection{Test case 2: Volumetric flow data compression}

The volumetric, stream-wise velocity $u$ field of the flow described in Section~\ref{sec:flow_compression_efficiency} is captured over 125 time steps to form the dataset of this experiment.
The $u$-velocity snapshots are collected for the entire volume, $(256 +2) \times (128 + 2) \times (128 +2)$ 3D spatial mesh (4360200 grid points), at each time level resulting in a total $125 \times4360200$ data matrix. The efficiency and accuracy of sub-sampled ID, SBR-SVD, and single-pass ID are analyzed for this test problem. Given the scale of the data matrix, analysis of the performance of the ID algorithm is omitted due to memory overflow on the machine used for the computations.

Figure~\ref{fig:flowreconstruct} presents a 3D (volumetric) visualization of the turbulent flow at a single time instance from the original simulation, as well as three reconstructed using SBR-SVD, sub-sampled ID, and single-pass ID. All three methods provide a qualitatively accurate reconstruction of the full simulation, though single-pass ID presents some noticeably different features at the small-scale level (e.g., near-wall turbulent structures) compared to those in the other three velocity fields. One should note, however, that the single-pass ID reconstruction is achieved with a much higher compression factor ---\thinspace 8 times higher to be precise\thinspace--- as the method compresses not only the time domain of the data but the spatial domain as well.

\begin{figure}[htb]
  \centering
  \includegraphics[width=0.49\textwidth]{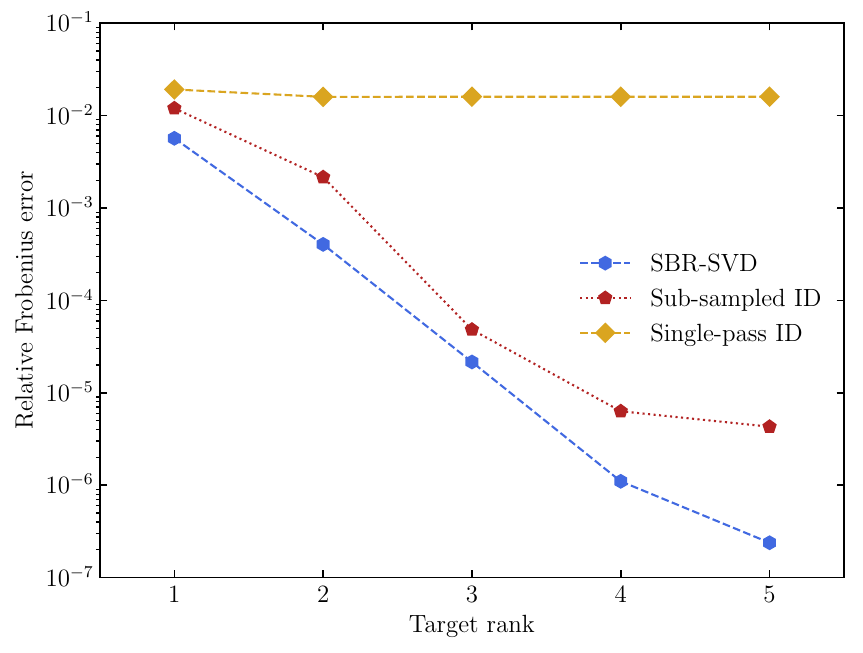}
  \includegraphics[width=0.49\textwidth]{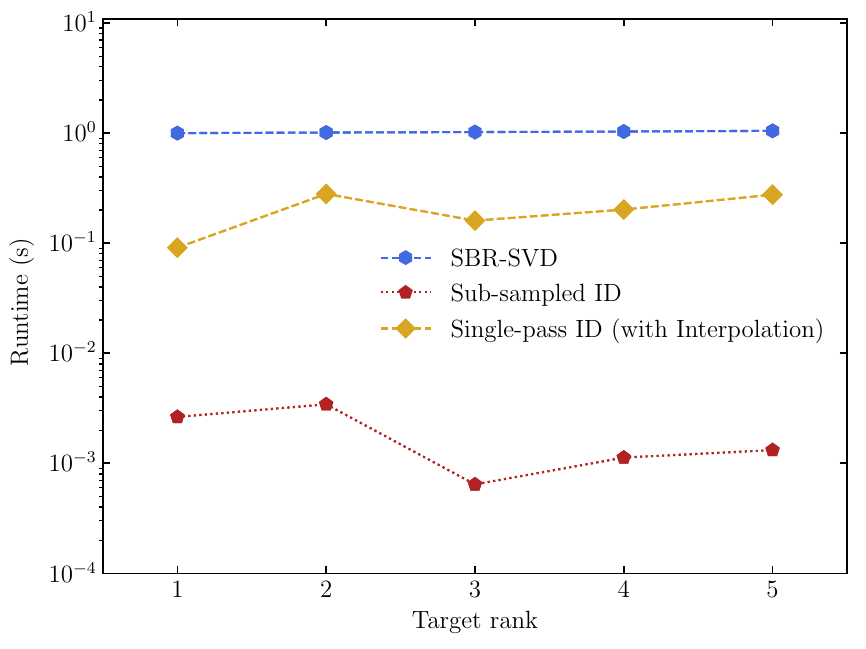}
  \caption{Left: Errors of three compression methods on volumetric data versus target rank. Right: Runtimes (normalized to the slowest trial of the SBR-SVD) of the three schemes versus target rank. In the right panel, the red line corresponds to the runtimes for both sub-sampled ID and single-pass ID {\it without} interpolation, while the yellow line reports the runtime of single-pass ID including the interpolation step.}	\label{fig:volerrorsruntimes}
\end{figure}

\begin{figure}[htb]
  \centering
  \includegraphics[width=0.65\textwidth]{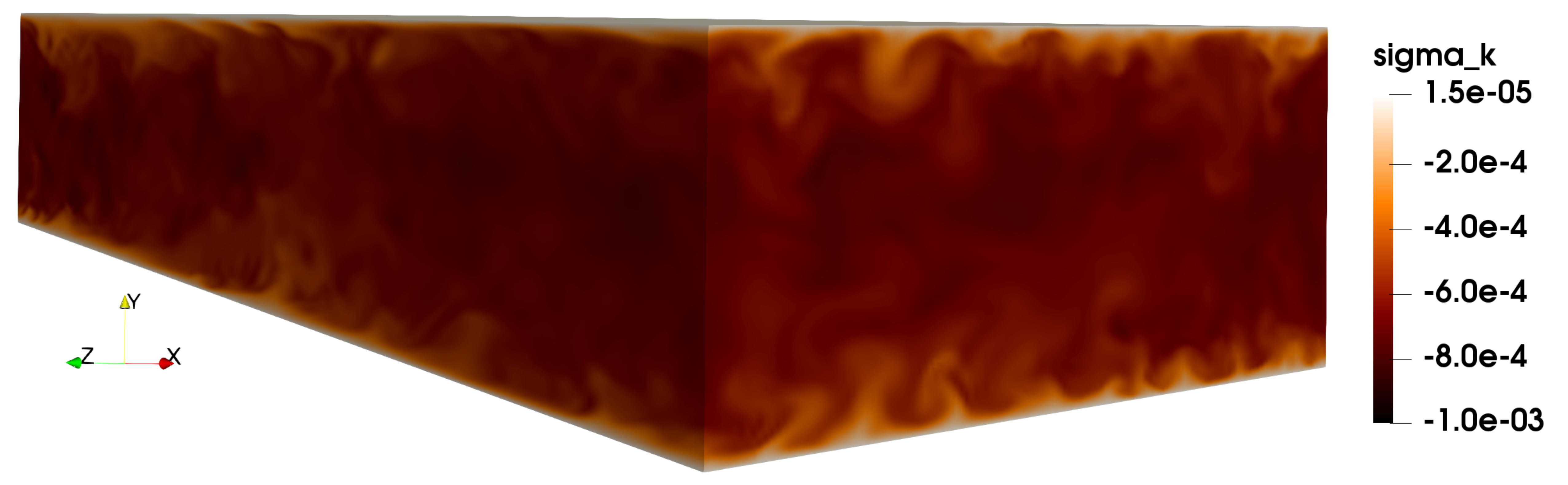}
  \caption{Singular vector associated with the largest singular value of the data matrix in Test case 2 computed via SBR-SVD. }	\label{fig:dominantmodechannel}
\end{figure}

The accuracy and runtime for the different methods are compared in Figure~\ref{fig:volerrorsruntimes}. Runtime is reported as the average over 20 independent trials per target rank value, and the oversampling parameter for SBR-SVD $p$ is set to 20. In single-pass ID, each dimension is sub-sampled by a factor of 2, leading to a spatial compression factor of 8, which compounds on the temporal compression factor achieved via low-rank approximation. The rank $k=1$ approximation for all three methods is roughly equally effective. For higher rank values, sub-sampled ID and SBR-SVD achieve far superior accuracy. Examining the singular vector corresponding to the largest singular value of the data matrix shown in Figure~\ref{fig:dominantmodechannel}, it can be seen that the dominant modes in the system are significantly complex in the spatial dimension. This ensures that the interpolation step in single-pass ID will cause significant loss of accuracy, even at the highest mode.

Runtime performances are analyzed relative to the slowest trial of SBR-SVD.
As depicted in Figure~\ref{fig:volerrorsruntimes}, sub-sampled ID is the fastest across all target rank values, further demonstrating its efficacy as a compression method for turbulent flow data. SBR-SVD is invariably the slowest of the three. Notably, single-pass ID is faster than SBR-SVD for all target ranks, even including the interpolation cost. This demonstrates that with sufficient optimizations, namely the incorporation of random projection prior to the column selection set, the additional cost of the interpolation step will not necessarily make single-pass ID slower than its SVD counterpart. 

Though sub-sampled ID is clearly more accurate than single-pass ID, it is important to remember that in an application where a simulation may take days to complete, or in which the data generated from such a simulation cannot be stored in RAM, a single-pass method, relative to a double-pass counterpart (e.g., sub-sampled ID) will significantly reduce the computational cost of simulation or loading time. Thus the trade-off in accuracy, as well as the additional runtime to interpolate back onto the full mesh ---\thinspace negligible relative to the cost of a high-fidelity simulation of a turbulent flow\thinspace--- may be worth it to practitioners faced with the dilemma of balancing accuracy and computational cost. 

\begin{figure}[htb!]
  \centering
  \includegraphics[width=0.49\textwidth]{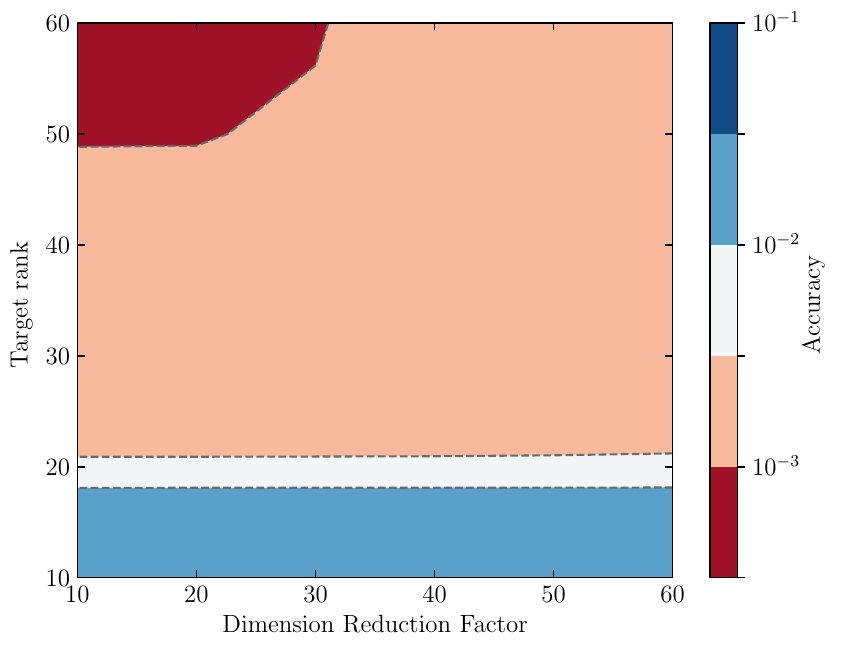}
  \includegraphics[width=0.48\textwidth]{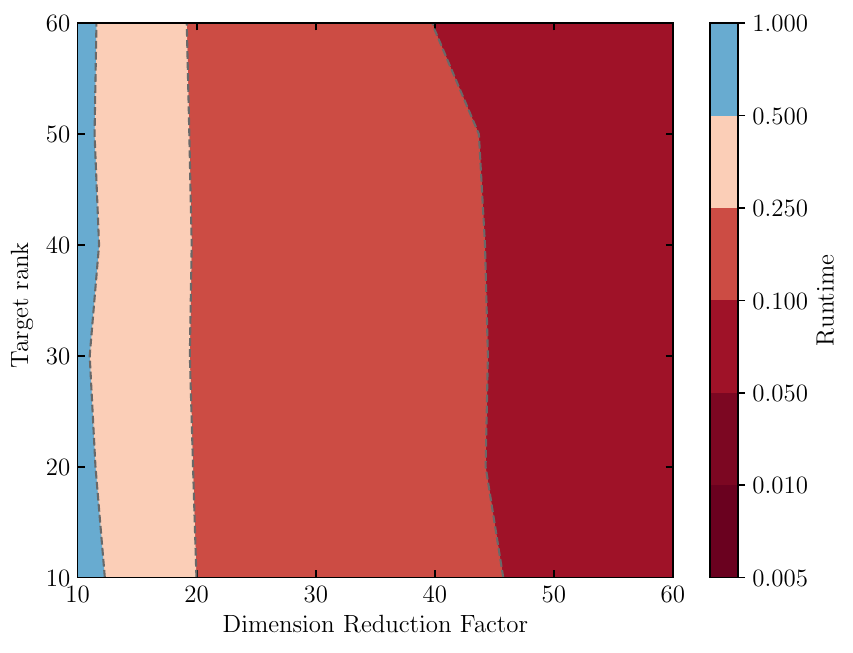}
  \includegraphics[width=0.49\textwidth]{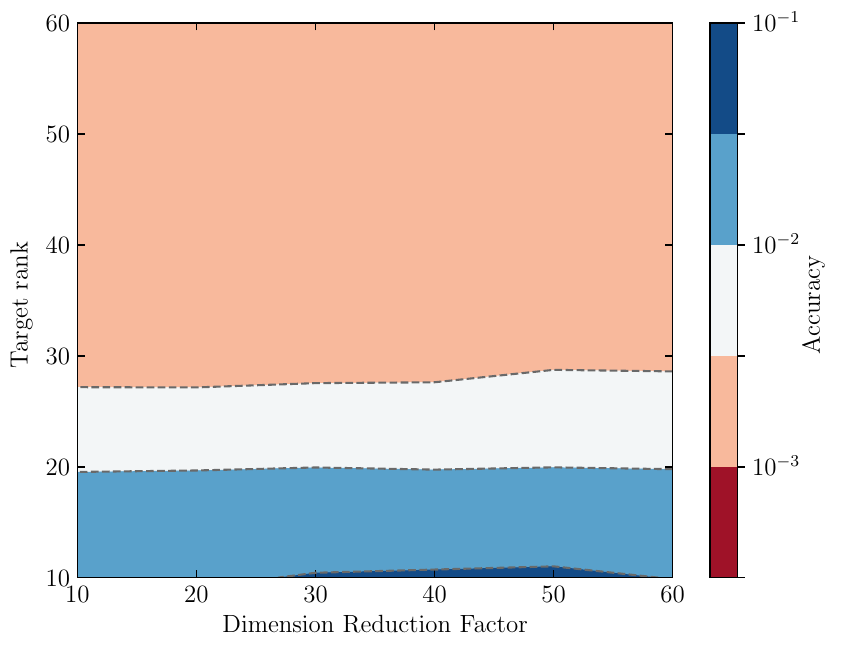}
  \includegraphics[width=0.48\textwidth]{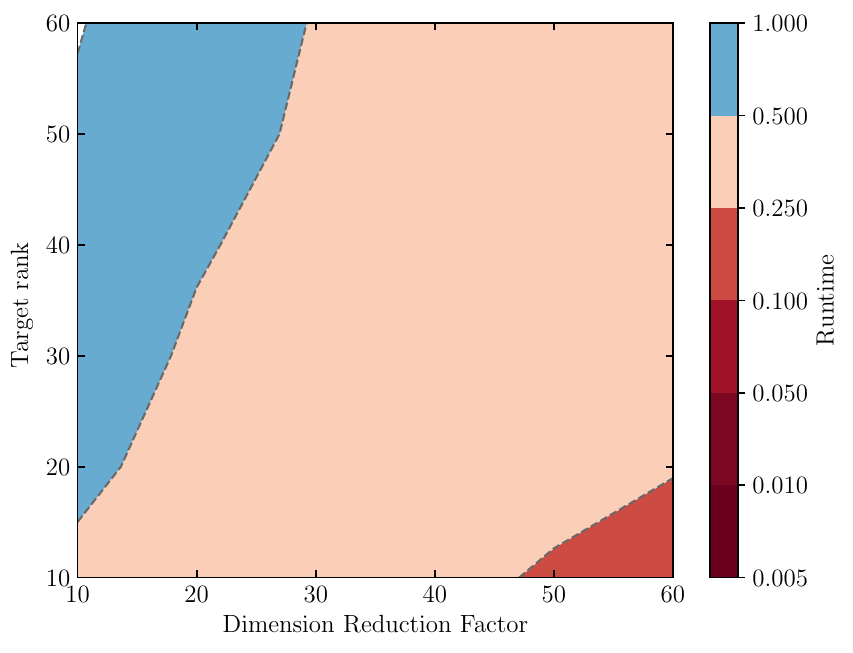}
  \includegraphics[width=0.49\textwidth]{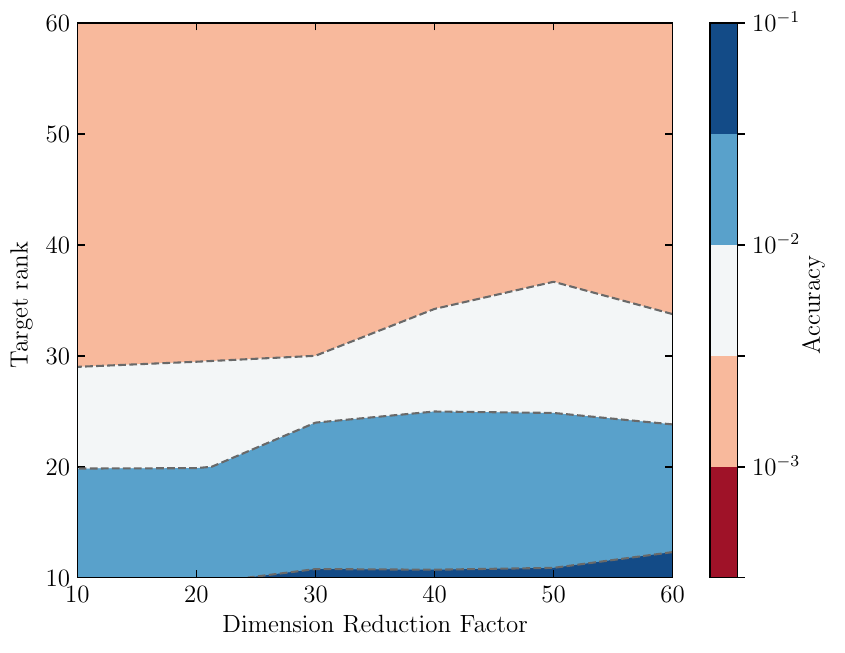}
  \includegraphics[width=0.48\textwidth]{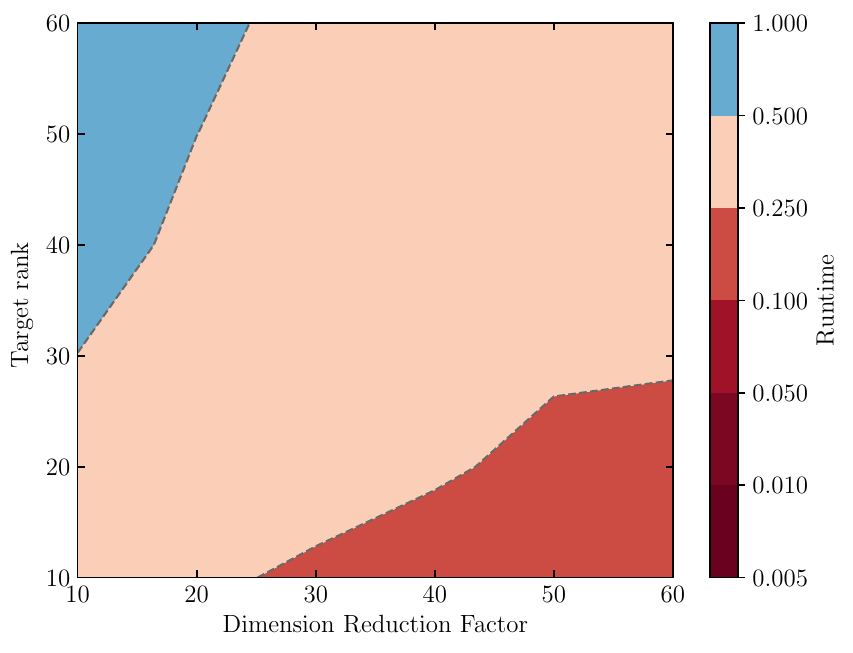}
  \caption{Accuracy (relative Frobenius error) and runtime (normalized by the largest value, i.e., slowest runtime: Randomized ID, dimension reduction factor of 10, target rank of 60) for individually traced particle $v$-velocity data at $St^+ = 1$. The dimension reduction factor is defined as the ratio $n/l$ for SBR-SVD and Randomized ID, and $n/n_c$ for sub-sampled ID. Top row: Accuracy (left) and runtime (right) of SBR-SVD. Center: Randomized ID. Bottom: Sub-sampled ID.}	\label{fig:runtime_accuracy}
\end{figure}

\subsection{Test case 3: Particle data compression efficiency and accuracy}
\label{sec:particle_compression_efficiency}

Wall-normal velocity, $v$, for $200000$ particles traced over 10000 time-steps is selected as test data. Three cases are considered in this section, corresponding to Stokes numbers $St^{+} = 0,1,10$. In particle-laden turbulent flows, the Stokes number characterizes the particle response time with respect to the flow time-scale.
Particles follow the turbulent motions at small $St^{+}$, whereas inertial forces dominate particle trajectories at large $St^{+}$.
As a result, particles at high Stokes numbers are less prone to multi-scale behavior, which leads to enhanced compression accuracy for the algorithms. 

SBR-SVD, ID, randomized ID (Algorithm~\ref{alg:randid}), and sub-sampled ID are applied to the individually traced particle data extracted from a particle-laden turbulent flow. In the randomized ID, random projection via multiplication by a matrix with i.i.d. Gaussian entries is used as opposed to the deterministic sub-sampling step in sub-sampled ID to reduce the dimension of the input data (see Step~\ref{step:randprojectID} in Algorithm~\ref{alg:randid}). Single-pass ID is not utilized in this section because the individually traced particle data lacks an associated mesh for interpolation, which single-pass ID requires.  

\begin{table}[htb]
\centering
\tabcolsep7pt\begin{tabular}{cccc}
\hline
\small{Stokes number} & \small{Method} & \small{Y-position data compression error} & \small{Runtime (s)} \\
\hline
& \small{SBR-SVD}& $3.6\cdot10^{-6}$ & \textcolor{white}{1} 1.3 \\
$St^+ = 0$ &\small{ID} &  $4.6\cdot10^{-6}$ & \textcolor{white}{1} 72.2 \\
& \small{Randomized ID}&  $6.0\cdot10^{-6}$ & \textcolor{white}{1} 4.6 \\
& \small{Sub-sampled ID}& $7.6\cdot10^{-6}$ & \textcolor{white}{5} 4.3 \\
\hline
& \small{SBR-SVD}&  $1.1\cdot10^{-6}$  & \textcolor{white}{1} 1.2 \\
$St^+ = 1$ & \small{ID} & $1.2\cdot10^{-6}$ & \textcolor{white}{1} 73.3\\
& \small{Randomized ID}&  $1.5\cdot10^{-6}$ & \textcolor{white}{1} 5.1 \\
& \small{Sub-sampled ID}&  $2.3\cdot10^{-6}$ & \textcolor{white}{5} 4.5 \\
\hline
& \small{SBR-SVD} & $2.6\cdot10^{-6}$ & \textcolor{white}{5} 1.2 \\
$St^+ = 10$ &\small{ID} & $2.7\cdot10^{-6}$ & \textcolor{white}{1} 82.5 \\
& \small{Randomized ID}&  $3.4\cdot10^{-6}$ & \textcolor{white}{5} 5.0 \\
& \small{Sub-sampled ID}&  $2.1\cdot10^{-5}$ & \textcolor{white}{5} 4.4 \\
\hline
\end{tabular}
\caption{Compression error and runtime achieved by the SBR-SVD, ID, randomized ID, and sub-sampled ID for $St^{+}=0, 1, 10$ $Y$-position data for target rank $100$ (temporal compression factor 100). Compression errors are computed for $10000\times 5000$ matrices of the particles wall-normal position in terms of relative Frobenius norm error. Runtimes are given in seconds. In sub-sampled ID, a sub-sampling factor of $n/n_c = 20$ is used. In randomized ID and SBR-SVD, the matrix is projected to the same dimension as in sub-sampled ID, using a Gaussian sampling matrix instead. In SBR-SVD, the QR decomposition is computed in blocks of size 110.}	\label{tab:stokesposition}
\end{table}

\begin{table}[htb!]
\centering
\tabcolsep7pt\begin{tabular}{cccc}
\hline
\small{Stokes number} & \small{Method} & \small{$V$-velocity data compression error} & \small{Runtime (s)} \\
\hline
& \small{SBR-SVD}& $4.9\cdot10^{-3}$ & \textcolor{white}{1} 1.2  \\
$St^+ = 0$ &\small{ID} &  $6.0\cdot10^{-3}$ & \textcolor{white}{1} 72.7 \\
& \small{Randomized ID}&  $7.7\cdot10^{-3}$ & \textcolor{white}{1} 5.3 \\
& \small{Sub-sampled ID}& $9.4\cdot10^{-3}$ & \textcolor{white}{5} 4.5 \\
\hline
& \small{SBR-SVD}&  $6.2\cdot10^{-4}$  & \textcolor{white}{1} 1.3 \\
$St^+ = 1$ & \small{ID} & $8.2\cdot10^{-4}$ & \textcolor{white}{1} 73.4\\
& \small{Randomized ID}&  $1.0\cdot10^{-3}$ & \textcolor{white}{1} 4.7 \\
& \small{Sub-sampled ID}&  $1.7\cdot10^{-3}$ & \textcolor{white}{5} 4.1 \\
\hline
& \small{SBR-SVD} & $6.9\cdot10^{-5}$ & \textcolor{white}{5} 1.1 \\
$St^+ = 10$ &\small{ID} & $9.6\cdot10^{-5}$ & \textcolor{white}{1} 75.7 \\
& \small{Randomized ID}&  $1.2\cdot10^{-4}$ & \textcolor{white}{5} 5.0 \\
& \small{Sub-sampled ID}&  $4.7\cdot10^{-3}$ & \textcolor{white}{5} 4.4 \\
\hline
\end{tabular}
\caption{Compression error and runtime achieved by the SBR-SVD, ID, randomized ID, and sub-sampled ID for $St^{+}=0, 1, 10$ $V$-velocity data for target rank $100$ (temporal compression factor 100). Compression errors are computed for $10000\times 5000$ matrices of the particles $v$-velocities in terms of relative Frobenius norm error. Runtimes are given in seconds. In sub-sampled ID, a sub-sampling factor of $n/n_c = 20$ is used. In randomized ID and SBR-SVD, the matrix is projected to the same dimension as in sub-sampled ID, using a Gaussian sampling matrix instead. In SBR-SVD, the QR decomposition is computed in blocks of size 110.}	\label{tab:stokesvelocity}
\end{table}

In Table~\ref{tab:stokesposition}, compression error and runtime results are reported for the four methods for $y$-position data of $5000$ particles tracked over $10000$ time-steps. In all three Stokes regimes, SBR-SVD is the fastest and most accurate. ID is the slowest, while the second most accurate, and randomized ID and sub-sampled ID are the least accurate but achieve runtimes close to those of SBR-SVD. Notice that in Algorithm \ref{alg:subsampid} for sub-sampled ID, the standard ID is performed on the sketch matrix $\bm{A}_c$ (Step 4). We anticipate the runtime of the sub-sampled ID presented in Table~\ref{tab:stokesposition} and in the subsequent results will be improved by performing randomized ID on $\bm{A}_c$ instead. Moreover, randomized ID and sub-sampled ID achieve similar accuracy, except in the $St^+=10$ case, where randomized ID is more accurate. In Table~\ref{tab:stokesvelocity}, the same metrics are applied to access the compression performance of the algorithms on $v$-velocity data collected from the same $5000$ particles as in Table~\ref{tab:stokesposition}. The same relative performances are observed as before; the performance gap between sub-sampled ID and the other three methods is particularly pronounced when $St^+ = 10$. Again, SBR-SVD is the fastest and most accurate method, while ID is by far the slowest. 

In Figure~\ref{fig:runtime_accuracy}, accuracy and runtime (normalized by the largest value across all algorithms and test cases) of SBR-SVD, randomized ID, and sub-sampled ID algorithms are shown in two contour plots for the case of $St^{+}=1$ with $5000$ particles selected from the original $200000$. The dimension reduction factor given as the horizontal axes of the plots is defined as the ratio $n/l$, for randomized ID and SBR-SVD, and $n/n_c$ for sub-sampled ID. All three methods prove to be robust with respect to dimension reduction, which leads to savings in runtime with insignificant loss of accuracy. In this test case, SBR-SVD outperforms all other algorithms in terms of accuracy as well as in terms of runtime. 

In Table~\ref{tab:stokesposition20000}, compression error and runtime results are reported for the four methods for $y$-position data of $20000$ particles tracked over $10000$ time-steps. In all three Stokes regimes, sub-sampled ID is the fastest of the four methods, while SBR-SVD is the most accurate. ID is by far the slowest, and the much faster randomized ID and sub-sampled ID are the least accurate. Randomized ID and sub-sampled ID achieve similar errors across Stokes regimes, with randomized ID slightly more accurate in each of the three test cases, in particular $St^+ = 10$. In Table~\ref{tab:stokesvelocity20000}, the same metrics are used to access the performance of the algorithm applied to the compression of $v$-velocity data collected from the same $20000$ particles as in Table~\ref{tab:stokesposition20000}. In this case, sub-sampled ID is again the fastest and SBR-SVD the most accurate of the three methods for all Stokes regimes. The performance gap between sub-sampled ID and the other schemes is again significant when $St^+ = 10$, as was also observed in the $5000$ particle test case (see Table~\ref{tab:stokesvelocity}).

A key difference between the $5000$ particle and $20000$ particle test cases is that sub-sampled ID is faster than all three other methods in the $20000$ particle case, while it is slower than SBR-SVD in the $5000$ particle particle test case. This is due to the fact that in the $5000$ particle case (considered in Tables~\ref{tab:stokesposition} and~\ref{tab:stokesvelocity}), the spatial domain of the matrix is much smaller ($1/2$ the size of the temporal domain), limiting the speedup to be gained by sub-sampling the matrix -- instead of using random projection -- prior to compression. When the ratio of the number of particles to the number of timesteps is increased to $2$ as in the $20000$ particle test case reported in Tables~\ref{tab:stokesposition20000} and~\ref{tab:stokesvelocity20000}, sub-sampled ID offers greater speedup relative to the other schemes. More generally, sub-sampled ID is ideally suited to PDE data in which the size of the spatial domain greatly exceeds that of the temporal domain, as was observed in the results of Test Case 2.

\begin{table}[htb!]
\centering
\tabcolsep7pt\begin{tabular}{cccc}
\hline
\small{Stokes number} & \small{Method} & \small{Y-position data compression error} & \small{Runtime (s)} \\
\hline
& \small{SBR-SVD}& $3.0 \times 10^{-6}$  & \textcolor{white}{1} 4.3  \\
$St^+ = 0$ &\small{ID} &  $3.9 \times 10^{-6}$ & \textcolor{white}{1}  284.3 \\
& \small{Randomized ID} & $5.0 \times 10^{-6}$  & \textcolor{white}{1} 5.9 \\
& \small{Sub-sampled ID} & $6.7 \times 10^{-6}$ & \textcolor{white}{5} 4.0  \\
\hline
& \small{SBR-SVD}&  $5.8 \times 10^{-7}$   & \textcolor{white}{1} 4.2  \\
$St^+ = 1$ & \small{ID} & $6.4 \times 10^{-7}$ & \textcolor{white}{1} 319.2\\
& \small{Randomized ID}& $8.2 \times 10^{-7}$   & \textcolor{white}{1} 5.9  \\
& \small{Sub-sampled ID}& $1.0 \times 10^{-6}$   & \textcolor{white}{5} 4.0 \\
\hline
& \small{SBR-SVD} & $6.0 \times 10^{-7}$ & \textcolor{white}{5} 4.3 \\
$St^+ = 10$ &\small{ID} & $6.2 \times 10^{-7}$ & \textcolor{white}{1} 337.6\\
& \small{Randomized ID}& $8.0 \times 10^{-7}$  & \textcolor{white}{5} 5.9 \\
& \small{Sub-sampled ID}&  $8.3 \times 10^{-6}$ & \textcolor{white}{5}  4.0\\
\hline
\end{tabular}
\caption{Compression error and runtime achieved by the SBR-SVD, ID, randomized ID, and sub-sampled ID for $St^{+}=0, 1, 10$ $Y$-position data for target rank $100$ (temporal compression factor 100). Compression errors are computed for $10000\times 20000$ matrices of the particles $y$-positions in terms of relative Frobenius norm error. Runtimes are given in seconds. In sub-sampled ID, a sub-sampling factor of $n/n_c = 80$ is used. In randomized ID and SBR-SVD, the matrix is projected to the same dimension as in sub-sampled ID, using a Gaussian sampling matrix instead. In SBR-SVD, the QR decomposition is computed in blocks of size 110.}	\label{tab:stokesposition20000}
\end{table}

\begin{table}[htb!]
\centering
\tabcolsep7pt\begin{tabular}{cccc}
\hline
\small{Stokes number} & \small{Method} & \small{V-velocity data compression error} & \small{Runtime (s)} \\
\hline
& \small{SBR-SVD}& $5.1 \times 10^{-3}$ & \textcolor{white}{1} 4.2\\
$St^+ = 0$ &\small{ID} & $6.4 \times 10^{-3}$  & \textcolor{white}{1} 331.4 \\
& \small{Randomized ID}&  $8.2 \times 10^{-3}$ & \textcolor{white}{1} 5.8 \\
& \small{Sub-sampled ID}& $1.0 \times 10^{-2}$ & \textcolor{white}{5} 4.1 \\
\hline
& \small{SBR-SVD}&  $4.4 \times 10^{-4}$  & \textcolor{white}{1}  4.6 \\
$St^+ = 1$ & \small{ID} & $5.9 \times 10^{-4}$ & \textcolor{white}{1} 316.5\\
& \small{Randomized ID}&  $7.6 \times 10^{-4}$ & \textcolor{white}{1} 6.4 \\
& \small{Sub-sampled ID}& $9.2 \times 10^{-4}$ & \textcolor{white}{5}  4.4\\
\hline
& \small{SBR-SVD} & $1.4 \times 10^{-4}$  & \textcolor{white}{5} 4.9 \\
$St^+ = 10$ &\small{ID} & $2.9 \times 10^{-4}$  & \textcolor{white}{1} 344.1 \\
& \small{Randomized ID}&  $3.9 \times 10^{-4}$  & \textcolor{white}{5}  6.1 \\
& \small{Sub-sampled ID}&  $7.3 \times 10^{-3}$  & \textcolor{white}{5} 4.1  \\
\hline
\end{tabular}
\caption{Compression error and runtime achieved by the SBR-SVD, ID, randomized ID, and sub-sampled ID for $St^{+}=0, 1, 10$ $V$-velocity data for target rank $100$ (temporal compression factor 100). Compression errors are computed for $10000\times 20000$ matrices of the particles $v$-velocities in terms of relative Frobenius norm error. Runtimes are given in seconds. In sub-sampled ID, a sub-sampling factor of $n/n_c = 80$ is used. In randomized ID and SBR-SVD, the matrix is projected to the same dimension as in sub-sampled ID, using a Gaussian sampling matrix instead. In SBR-SVD, the QR decomposition is computed in blocks of size 110.}	\label{tab:stokesvelocity20000}
\end{table}

Next, the compression accuracy of SBR-SVD, randomized ID, and sub-sampled ID applied to the reconstruction of time-averaged wall-normal particle concentrations and velocities for the full 200000 particle data is investigated, depicted in Figure~\ref{fig:particle_laden_qoi}. The $10000 \times 200000$ matrix is first separated into blocks of size $ 10000 \times 5000$. Then, each block is compressed individually. Empirically, this was found to enhance the compressibility of the dataset.
SBR-SVD, randomized ID and sub-sampled ID are the three methods considered in this case, with the target rank for each block set to $k=100$ in all test cases. The simulation in question consists of 10000 timesteps, therefore, all of the algorithms attain temporal compression factors of 100. The accuracy in compressing particle positions is quantified by the normalized concentration (with respect to the centerline of the channel), $C/C_{cl}$, while the compression of particles velocity, $v_{p}^{+}=v_{p}/u_{\tau}$, is evaluated by means of scaled probability density functions (PDF) in the viscous sublayer ($y^+$ $<$ 5), buffer layer (5 $<$ $y^+$ $<$ 30), and log-law region ($y^+$ $>$ 30). The results shown in the plots demonstrate that SBR-SVD, randomized ID, and sub-sampled ID are able to compress the Lagrangian data without significant loss of reconstruction accuracy.

\begin{figure}[htb!]
  \centering
  \includegraphics[width=0.48\textwidth]{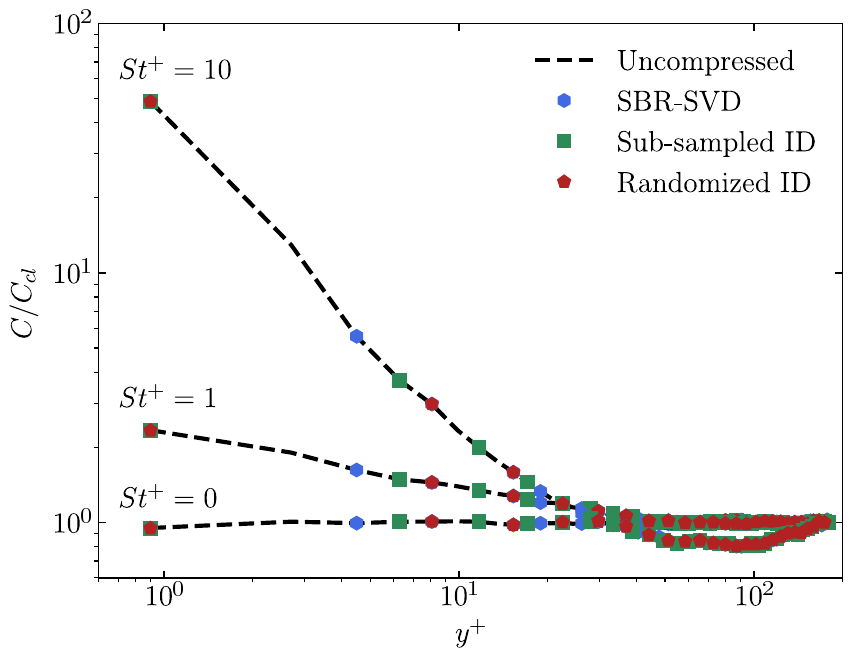}
  \includegraphics[width=0.49\textwidth]{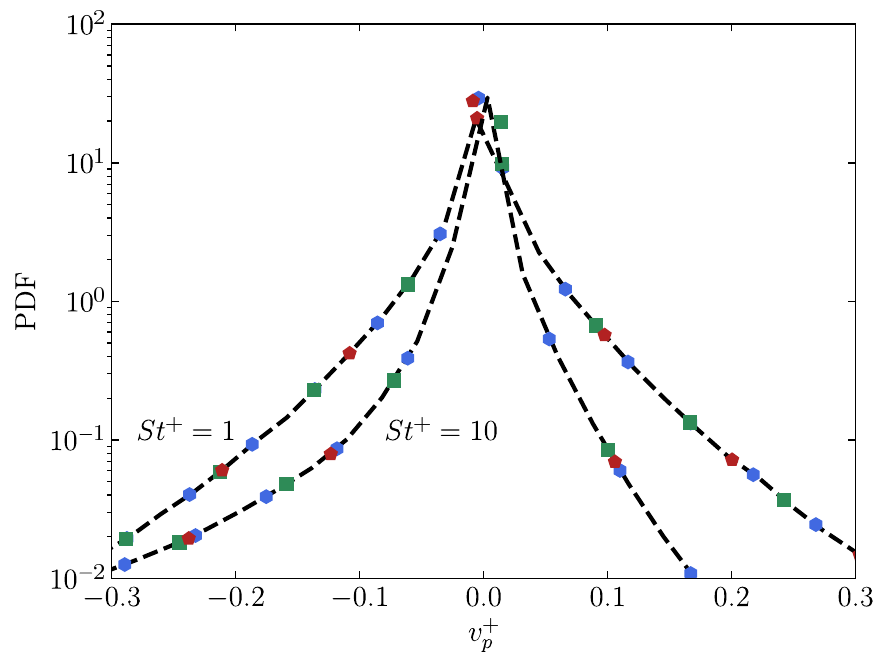}
  \includegraphics[width=0.49\textwidth]{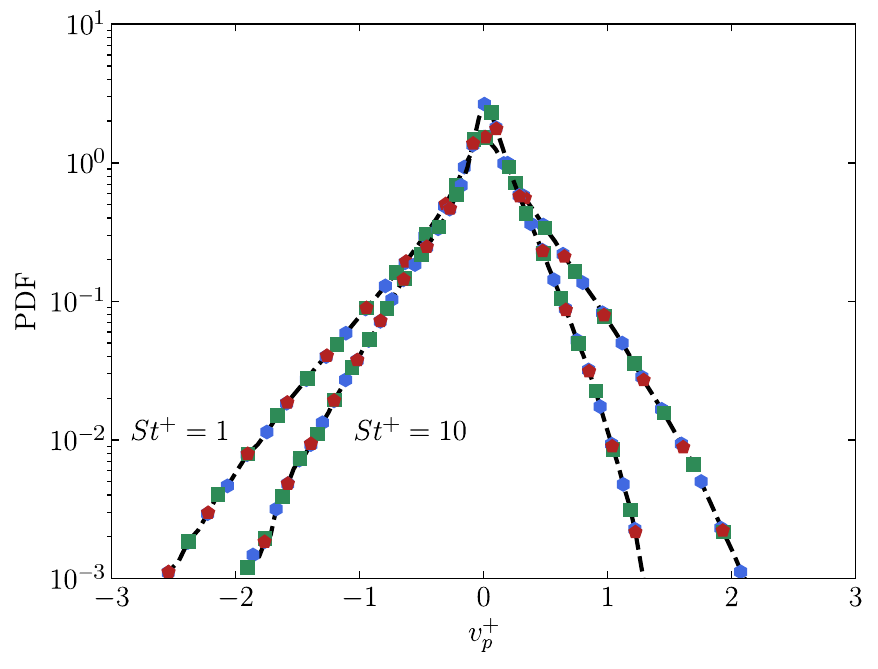}
  \includegraphics[width=0.49\textwidth]{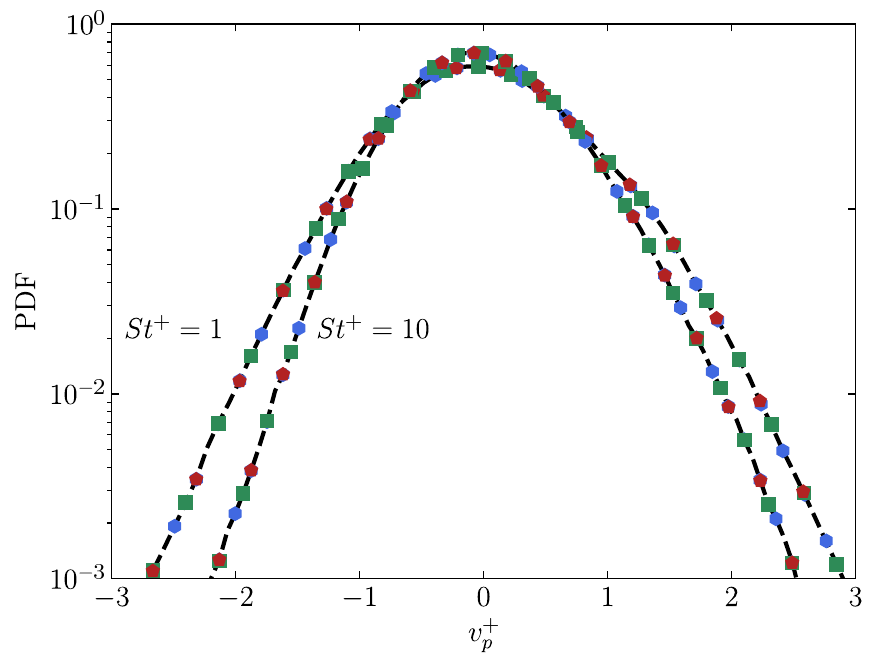}
  \caption{Top left: Steady-state mean normalized particle concentration in the wall-normal direction (wall-units). Top right: Scaled PDF of wall-normal particle velocity (wall-units) in the viscous sublayer ($y^+ < 5$). Bottom left: Scaled PDF of wall-normal particle velocity (wall-units) in the buffer layer ($5 < y^+ < 30$). Bottom right: Scaled PDF of wall-normal particle velocity (wall-units) in the log-law region ($y^+ > 30$). For SBR-SVD, sub-sampled ID, and randomized ID, results are recovered from compressed data with temporal compression factor $100$, corresponding to a target rank of $k=100$ used on each $10000 \times 5000$ matrix block.}	\label{fig:particle_laden_qoi}
\end{figure}

\section{Conclusions}	
\label{sec:conclusions}

Inspired by growing issues in memory capacity and I/O emerging due to drastically increased floating-point performance on modern high-performance computers, pass-efficient matrix decomposition algorithms for the utility of compression of large-scale streaming PDE data are investigated in this work.  Within the broad field of matrix decomposition methods, this work focused on four algorithms: a single-pass randomized SVD-based algorithm (SBR-SVD), and three variants of interpolative decomposition, namely, randomized ID, (grid) sub-sampled ID, and (grid sub-sampled) single-pass ID. Prior to this work, all existing ID algorithms required at least two passes over the input data. The novel single-pass ID addresses this gap within this class of matrix factorization. At the core of the new ID methods, sub-sampled ID and single-pass ID, is the utility of a sketch based on coarse description of the data associated with the computational grid which enables a faster construction of ID. Bounds on the approximation error of these variants of ID are derived and insight into the role of the coarse grid sketch is presented.

On the empirical side, these low-rank factorization techniques are employed for the compression of two turbulent channel flows, one unladen and another laden with particles. Planar, volumetric, and discrete particle data are considered. In all cases the data admits low-rank representation and these methods are able to exploit this structure to achieve significant data compression. Among the one-pass (streaming) algorithms in the unladen turbulent channel flow test case, SBR-SVD is typically faster and more accurate than single-pass ID, but both are effective in the compression of the outlet velocity dataset as well as the volumetric dataset. Among the double-pass ID algorithms, sub-sampled ID performs as well as the ID in the recovery of flow statistics in the first two test cases, indicating promise for its application in other problems with similar low-rank structure and smooth spatial variation. SBR-SVD and sub-sampled ID feature significant reductions in runtime with insignificant losses in accuracy as compared to ID in all test cases. Individually traced particle data, lacking an associated mesh, remains an application to which single-pass ID is not adapted. In situations where the size of the spatial grid or the number of particles is large (relative to the number of time snapshots), we observe faster runtime achieved by sub-sampled ID as compared with randomized ID or SBR-SVD, thus justifying the utility of sub-sampled ID in high-dimensional data compression.
 
Despite their generally inferior performance in terms of runtime and accuracy on the three test cases considered in this work, the ID algorithms have  key benefits over SBR-SVD. One is that all three variants of ID examined in this work may take as input an error tolerance. As a result, the three ID algorithms do not require \textit{a priori} knowledge of the compressibility, i.e., the numerical rank, of a dataset. Work has been done recently to develop randomized SVD algorithms which adaptively determine the rank of a dataset~\cite{ji2016rank}, though these methods have yet to be developed in a single-pass framework to the best of the authors' knowledge. ID also generates an intuitive decomposition of a matrix, as it expresses the row space of a matrix in a basis of a subset of the rows themselves. This makes it a useful tool for domain experts who may find the SVD more difficult to interpret in the context of their raw data. For example, in a turbulent flow simulation whose time realizations arrive in RAM as row vectors, row ID can be interpreted as identifying the {\it most important} snapshots of the entire solution, in the sense that they form a near-optimal self-expressive basis to reconstruct the data. The SVD, on the other hand, does not provide exact snapshot reconstruction at any point in the temporal domain. 

An ongoing extension of this work is the development of scalable, parallel implementations of these methods directly withing a flow solver, with the aim to enable {\it in situ} compression. We anticipate some of the observations made regarding the relative runtime of the compression algorithms considered here may change when executed in parallel. Other future avenues of research directions include the utility of i) higher order interpolation schemes (potentially on unstructured grids) to improve the accuracy of single-pass ID and ii) spatially compressing the row skeleton generated by the ID using, e.g., compressed sensing algorithms~\cite{baraniuk2010model,jokar2010sparse} to further enhance current compression factors.

\section*{Acknowledgments} 

This work was funded by the United States Department of Energy's National Nuclear Security Administration under the Predictive Science Academic Alliance Program (PSAAP) II at Stanford University, Grant DE-NA-0002373. A.D. acknowledges funding by the US Department of Energy's Office of Science Advanced Scientific Computing Research, Award DE-SC0006402, and National Science Foundation, Grant CMMI-145460. L.J. acknowledges funding by the Beatriz Galindo Program (Distinguished Researcher, BGP18/00026) of the Ministerio de Ciencia, Innovaci\'on y Universidades, Spain.

An award of computer time was provided by the ASCR Leadership Computing Challenge program. This research used resources of the Argonne Leadership Computing Facility, which is a Department of Energy's Office of Science User Facility supported under contract DE-AC02-06CH11357. This research also used resources of the Oak Ridge Leadership Computing Facility, which is a Department of Energy's Office of Science User Facility supported under contract DE-AC05-00OR22725. A.D. acknowledges the fruitful discussions with Prof. Kenneth Jansen and Stephen Becker from CU Boulder.

\bibliography{main}
\bibliographystyle{model1-num-names}

\end{document}